\documentclass[11pt]{article}
\usepackage[letterpaper, margin=1in]{geometry}
\usepackage{thm-restate}
\usepackage{makecell}
\usepackage{amsfonts}
\usepackage{amsmath}
\usepackage{amsthm}
\usepackage{graphicx}
\usepackage{xcolor}
\usepackage[colorlinks=true, allcolors=blue]{hyperref}

\usepackage{algorithm}
\usepackage{algpseudocode}
\usepackage{tcolorbox}

\newtheorem{defn}{Definition}
\newtheorem{oq}{Open Question}
\newtheorem{prop}{Proposition}
\newtheorem{lem}{Lemma}
\newtheorem{thm}{Theorem}

\DeclareMathOperator{\poly}{poly}
\DeclareMathOperator{\polylog}{polylog}
\usepackage{cite}
\usepackage{algorithm}
\usepackage{algpseudocode}

\newcommand{\E}{\mathop{\textbf{E}}}

\title{Breaking Barriers for Distributed MIS by Faster Degree Reduction}

\author{
  Seri Khoury\thanks{INSAIT, Sofia University “St. Kliment Ohridski”. \texttt{seri.khoury@insait.ai}} 
  \and 
  Aaron Schild\thanks{Google Research. \texttt{aaron.schild@gmail.com}}
}

\date{}

\begin{document}
\maketitle
\begin{center}

\end{center}
\begin{abstract}

We study the problem of finding a maximal independent set (MIS) in the standard LOCAL model of distributed computing. Classical algorithms by Luby [JACM'86] and Alon, Babai, and Itai [JALG'86] find an MIS in \(O(\log n)\) rounds in \(n\)-node graphs with high probability. Despite decades of research, the existence of \emph{any} \(o(\log n)\)-round algorithm for general graphs remains one of the major open problems in the field.

Interestingly, the hard instances for this problem must contain constant-length cycles. This is because there exists a sublogarithmic-round algorithm for graphs with super-constant girth; i.e., graphs where the length of the shortest cycle is \( \omega(1) \), as shown by Ghaffari~[SODA'16]. Thus, resolving this \(\approx 40\)-year-old open problem requires understanding the family of graphs that contain \(k\)-cycles for some constant \(k\).

In this work, we come very close to resolving this \(\approx 40\)-year-old open problem by presenting a sublogarithmic-round algorithm for graphs that can contain \(k\)-cycles for all \(k > 6\). Specifically, our algorithm finds an MIS in 
\(
O\left(\frac{\log \Delta}{\log(\log^* \Delta)} + \mathrm{poly}(\log\log n)\right)
\)
rounds, as long as the graph does not contain cycles of length \(\leq 6\), where $\Delta$ is the maximum degree of the graph. As a result, we push the limit on the girth of graphs that admit sublogarithmic-round algorithms from $k = \omega(1)$ all the way down to a small constant $k=7$. 
 Moreover, our result has the following further implications:

\begin{itemize}
    \item \textbf{Refuting a conjecture about MIS in trees:} By combining our algorithm with a low-arboricity-to-low-degree reduction by Barenboim, Elkin, Pettie, and Schneider [JACM'16], we achieve an
\(
O\left(\sqrt{\frac{\log n}{\log(\log^* n)}}\right)
\)
-round algorithm in trees. This refutes a conjecture in the book by Barenboim and Elkin that finding an MIS in trees requires \(\Theta(\sqrt{\log n})\) rounds. 
    \item \textbf{Separating MIS from Maximal Matching in trees:} Together with a very recent work that shows a $\Omega(\sqrt{\log n})$ lower bound for Maximal Matching (MM) in trees, our result implies a surprising and counterintuitive separation between MIS and MM in trees. While MM can only be easier than MIS in general graphs, it becomes strictly \emph{harder} in trees. This also implies that MIS itself is strictly harder to solve in general graphs than in trees.
\end{itemize}

\end{abstract}

\newpage
\section{Introduction}

One of the earliest and most extensively studied problems in distributed graph algorithms is the problem of finding a maximal independent set (MIS)\cite{Luby86,AlonBI86,RozhonG20,2013Barenboim,BalliuBHORS21,Ghaffari16,BarenboimEPS16,KuhnMW16,BarenboimE10,Linial92,Naor91,0001GHIR23,0001G23,GhaffariP19,ecomposition,BalliuBO22,Balliu0KO22,Balliu0KO21,LenzenW11,MetivierRSZ11,wattenhofer2020mastering,SchneiderW08}. Its origins can be traced back to the 1980s, coinciding with the introduction of the standard LOCAL model~\cite{Linial92,Luby86,AlonBI86}. In this model, the input graph is abstracted as a network of $n$ nodes that can communicate with each other via synchronized rounds. In each round, each node can send unbounded-size messages to its neighbors and perform some local computation. The goal is for the nodes to collaboratively solve a graph problem (e.g., find an MIS) while minimizing the number of communication rounds.

Classical algorithms by Luby~\cite{Luby86} and Alon, Babai, and Itai~\cite{AlonBI86} find an MIS in \( O(\log n) \) rounds in the LOCAL model with high probability.\footnote{Throughout the paper, we say that an algorithm succeeds with high probability if it succeeds with probability. \( 1 - 1/n^c \) for arbitrarily large constant \( c > 1 \).}  Remarkably, despite decades of research, these algorithms remain the best known for general graphs. The question of whether there is a sublogarithmic-round algorithm remains one of the major and earliest open questions in the field~\cite{2013Barenboim,Ghaffari16,BarenboimEPS16,HarrisSS18,wattenhofer_podc,abs-2406-19430}.

This question is pointed out in the book by Barenboim and Elkin~\cite{2013Barenboim}, as well as in lecture notes by Wattenhofer, who states that ``one of the main open problems in distributed computing is whether one can beat this logarithmic time..."~\cite{wattenhofer_podc}. It is further highlighted in several seminal papers in the field. Barenboim, Elkin, Pettie, and Schneider~\cite{BarenboimEPS16} stress that ``the most difficult problem is to find optimal... algorithms for MIS and maximal matching, or, as a first step, \emph{any} \( o(\log n) \) time algorithm". It was also mentioned in the first sublogarithmic-round algorithm for distributed coloring by Harris, Schneider, and Su~\cite{HarrisSS18}, as well as implicitly by Ghaffari~\cite{Ghaffari16}, who mentions that the gap between the upper and lower bounds for MIS is ``perhaps the most interesting."

%\seri{rephrase the open question and match the theorems}

\begin{oq}(\cite{2013Barenboim,BarenboimEPS16,Ghaffari16,HarrisSS18,wattenhofer_podc,abs-2406-19430})\label{oq:main}\label{trees}
    Is there a randomized distributed algorithm that finds an MIS in $o(\log n)$ rounds with high probability?
\end{oq}

On the other hand, until very recently, the best-known lower bound for randomized algorithms as a function of \( n \) was the \(21\)-year-old \(\Omega\left(\sqrt{\frac{\log n}{\log \log n}}\right)\) lower bound by Kuhn, Moscibroda, and Wattenhofer~\cite{KuhnMW16}. When expressed as a function of both the maximum degree \( \Delta \) and \( n \), this bound becomes \(\Omega\left(\min\left\{\sqrt{\frac{\log n}{\log \log n}}, \frac{\log \Delta}{\log \log \Delta}\right\}\right)\). In a very recent concurrent work, an improved lower bound of $\Omega\left(\min\{\sqrt{\log n},\log\Delta\}\right)$ was presented~\cite{LBMatching}.
This dependence on \( \Delta \) is tight for small values of \( \Delta \ll \exp\left(\sqrt{\log n}\right) \)~\cite{Ghaffari16} (as we explain shortly). For even smaller values of \( \Delta \ll \log \log n \), this dependency grows further due to the \(\Omega\left(\min\left\{\Delta, \frac{\log \log n}{\log \log \log n}\right\}\right)\) lower bound by Balliu, Brandt, Hirvonen, Olivetti, Rabie, and Suomela~\cite{BalliuBHORS21}, which is also tight~\cite{BarenboimE09,Kuhn09,BarenboimEK14}.

All these lower bounds apply to randomized algorithms. For deterministic algorithms, the goal is to approach the \(O(\log n)\) state-of-the-art complexity achieved by randomized ones.\footnote{The current best upper bound for deterministic algorithms is roughly \(O(\log^{5/3} n)\) rounds~\cite{ecomposition}, and the current best lower bound is \(\Omega(\min\{\Delta, \log n / \log \log n\})\)~\cite{BalliuBHORS21}. For more on deterministic algorithms, we refer the reader to the excellent surveys by Suomela~\cite{Suomela13} and Rozhon~\cite{abs-2406-19430}.} In this work, aiming to advance our understanding of Open Question~\ref{oq:main}, we focus on the randomized complexity of MIS as a function of \( n \), where the gap between the upper and lower bounds remains at \(\Theta\left(\sqrt{\log n}\right)\). As a result of our approach, we also obtain a new result as a function of \( \Delta \).

\paragraph{Recent Progress:} Over the past one and a half decades, two major breakthroughs have significantly advanced our understanding of Open Question~\ref{oq:main}. Barrenboim, Elkin, Pettie, and Schneider~\cite{BarenboimEPS16} presented an algorithm that takes \(O(\log^2\Delta + 2^{\sqrt{\log\log n}})\) rounds, which implies a $o(\log n)$-round algorithm for graphs with \(\Delta\in 2^{o(\sqrt{\log n})}\).

Later, Ghaffari~\cite{Ghaffari16} improved this result by presenting an \(O(\log\Delta + 2^{\sqrt{\log\log n}})\)-round algorithm, which was further improved to \(O(\log\Delta + \poly(\log\log n))\) rounds by Rozhon and Ghaffari~\cite{RozhonG20}.\footnote{The improvement from $2^{O(\sqrt{\log\log n})}$ to $\poly(\log\log n)$ in the additive terms comes from a recent breakthrough in deterministic algorithms~\cite{Luby86,AlonBI86}. In fact, these additive terms are simply $det(\log n)$, which is the deterministic complexity of MIS in $(\log n)$-node graphs. For a long time, $det(n')$ was only known to be $2^{O(\sqrt{\log n'})}$ until it was improved to $\poly\log n'$ by~\cite{RozhonG20}.} Hence, we now know that we can essentially replace the logarithmic dependence on \(n\) from classical algorithms with a logarithmic dependence on \(\Delta\). Yet, for graphs with \(\Delta = n^{\Omega(1)}\), the $\approx 40$-year-old \(O(\log n)\) bound is still the best known.

\paragraph{The Degree Reduction Barrier:} Trying to resolve Open Question~\ref{oq:main} affirmatively by introducing a better upper bound encounters a common barrier faced by all prior techniques.
 This is the \emph{degree reduction barrier}. Intuitively speaking, all prior algorithms follow a common theme: first, we find an independent set $\mathcal{I}$, then we remove the nodes in $\mathcal{I}$ from the graph together with their neighbors, and we repeat until the graph is empty. To show that the algorithm terminates after a desired number of rounds, the independent set $\mathcal{I}$ needs to be selected in a way that helps ``make progress" in the resulting graph; i.e., the removal of $\mathcal{I}$ from the graph ideally simplifies it to some extent. 

In the classical algorithms by~\cite{Luby86,AlonBI86}, they show a one-round algorithm that finds an independent set whose removal together with the neighboring nodes deletes a constant fraction of the edges in expectation. Therefore, after $O(\log n)$ rounds, the algorithm finds an MIS with high probability. The later algorithms by~\cite{BarenboimEPS16} and~\cite{Ghaffari16} show that instead of making progress by reducing the number of remaining edges, one can make progress by reducing the maximum degree $\Delta$. 

In~\cite{BarenboimEPS16}, the authors show that invoking \(O(\log\Delta)\) rounds of the classical algorithms by~\cite{Luby86,AlonBI86} reduces the degree of a typical node by a factor of 2. Therefore, after \(O(\log^2\Delta)\) rounds, a typical node is deleted from the graph. For the remaining (non-typical) nodes, they show that the size of each remaining connected component is sufficiently small. This is also known as the \emph{shattering phenomenon}. Once the graph shatters into small connected components, we can complete the MIS in these components efficiently.

Later, in~\cite{Ghaffari16}, the author showed an algorithm that reaches the shattering phenomena substantially faster -- in only $O(\log \Delta)$ rounds. In this algorithm, each node spends $O(\log \Delta)$ rounds in which it is removed with constant probability. One way to think about this algorithm is that a typical node experiences $O(\log\Delta)$ degree reductions, each by a constant factor. 

To achieve a sublogarithmic running time for MIS, we need to devise an algorithm that reduces the degree of a node even faster. Such an algorithm is particularly necessary in graphs with \(\Delta = n^{\Omega(1)}\). However, all prior techniques reduce the degree of a node by at most a constant factor per round. This is exactly the degree reduction barrier. We need novel techniques to find an independent set in one round (or in a constant number of rounds) whose removal reduces the degree of a typical node by a super-constant factor.

\paragraph{Faster Degree Reduction in Sparse Graphs:} Notably, the degree reduction barrier does not persist in graphs that admit a certain notion of sparsity; namely, low-\emph{arboricity} graphs. Roughly speaking, the arboricity \( \lambda \) of a graph \( G \) is the maximum density of any subgraph \( H \) of \( G \).\footnote{Formally, \(\lambda\) is defined as \(\max_{H \subseteq G} \frac{|E(H)|}{|V(H)| - 1}\). Note that the arboricity $\lambda\leq \Delta$.} In~\cite{BarenboimEPS16}, the authors show that one can reduce the degree of the graph from $\Delta$ to $\lambda\cdot 2^{\sqrt{\log n}}$ in only $O(\sqrt{\log n})$ rounds. Therefore, by combining this with the algorithm of~\cite{Ghaffari16}, one obtains an algorithm that finds an MIS in \( O\left(\log \lambda + \sqrt{\log n}\right) \) rounds. Consequently, this extends the family of graphs that admit a sublogarithmic-round algorithm to all graphs with \( \lambda \in  n^{o(1)} \).

Moreover, this also provides a sublogarithmic-round algorithm for all graphs that do not contain cycles of constant length. This is because, for any even integer \( k \in [4, \log n/\log\log n] \), any graph that lacks cycles of length at most \( k \) has arboricity at most \( O(k\cdot n^{1/k}) < n^{O(1/k)} \)~\cite{BONDY197497}. Thus, graphs in which the length of the shortest cycle (also known as the \emph{girth}) is super-constant have arboricity \( n^{o(1)} \), and therefore they admit a sublogarithmic-round algorithm. This leaves graphs that can contain constant-length cycles as the main hard instances toward resolving Open Question~\ref{oq:main}. (Here, by ``can contain,” we mean that the graph can contain arbitrarily many constant-length cycles. If the number of constant-length cycles is sufficiently small, or if they admit a specific structure, then one can obtain similar upper bounds on the arboricity in these special cases.~\cite{Jiang_2020})

Due to the centrality of MIS in distributed graph algorithms, and to make further progress on the $\approx 40$-year old Open Question~\ref{oq:main}, it is crucial to investigate whether we can break the degree reduction barrier in the family of graphs that can contain constant-length cycles.

\paragraph{The Constant Survival Probability Barrier:} When trying to overcome the degree reduction barrier, one faces a related and a much older barrier, which we refer to as the \emph{constant survival probability barrier}. Intuitively, for an algorithm to reduce the degree of a typical node by a super-constant factor in a round, it needs to delete a node with probability \(1 - o(1)\). In other words, the \emph{survival} probability of a node needs to be \(o(1)\). This way, a node would lose all but a \(o(1)\)-fraction of its neighbors in expectation. Observe that we do not need to achieve this guarantee in \emph{all} rounds; but it suffices to have it in \emph{most} of the rounds.\footnote{We note that it impossible to get sub-constant survival probability in all rounds. In fact, it is already impossible to get a constant survival probability in all rounds, as implied by the recent lower bound of~\cite{BalliuGKO23} for the node-average complexity of MIS. Roughly speaking, the node average complexity of an algorithm is the average (over the $n$ nodes) of the expected number of rounds it takes a node to arrive at its part of the solution.} 

However, achieving a sub-constant survival probability in a typical round is a classical barrier. In fact, the classical algorithms by~\cite{Luby86,AlonBI86} achieve a constant survival probability only in regular graphs, and it was only 30 years later, in the breakthrough by~\cite{Ghaffari16}, that an algorithm achieving constant survival probability in general graphs (in most of the \(O(\log\Delta)\) rounds) was found. 

This barrier persists even in trees, and it translates to an $O(\sqrt{\log n})$ barrier on the running time of MIS in trees. In particular, the state-of-the-art $O(\sqrt{\log n})$ upper bound for MIS in trees is also implied by the algorithm of~\cite{Ghaffari16} that achieves constant survival probability in most of the rounds. Getting better than a constant survival probability in a typical round in trees would imply a $o(\sqrt{\log n})$-round algorithm, which would break a conjecture in the book by Barenboim and Elkin~\cite{2013Barenboim}, and would imply a surprising separation between MIS and Maximal Matching in trees when combined with the recent lower bound of~\cite{LBMatching} (we elaborate on both implications shortly). 

\paragraph{Our Results:} In this work, we come very close to resolving the $\approx$40-year-old Open Question~\ref{oq:main} by presenting a sublogarithmic-round algorithm for graphs without cycles of length \(\leq 6\). Specifically, for such graphs, we devise a novel two-round algorithm in which the survival probability of a typical node is \(o(1)\), allowing us to overcome both the constant survival probability and degree reduction barriers.

%\seri{fix the annoying overflow in the theorem's box}

\begin{restatable}{thm}{thmmain}
\label{thm:main-girth}
There is a randomized distributed algorithm that, with high probability, finds an MIS in $O\left(\frac{\log\Delta}{\log(\log^*\Delta)} + \poly(\log\log n)\right)$ rounds in graphs without cycles of length $\leq 6$.
\end{restatable}

This is the first sublogarithmic-round algorithm for graphs that can contain arbitrarily many constant-length cycles, provided that these cycles are of length greater than $6$. 
 As a result, we push the limit on the girth of graphs that admit sublogarithmic-round algorithms from $k = \omega(1)$ all the way down to a small constant $k=7$. 
 Consequently, in order to fully resolve Open Question~\ref{oq:main}, it remains to solve the instances that can contain cycles of length at most $6$, which is an extremely local property. 
 
 Observe that while the recent \(\Omega(\min\{\log \Delta, \sqrt{\log n}\})\) lower bound by~\cite{LBMatching} implies that the exact runtime from Theorem~\ref{thm:main-girth} cannot be achieved for low-degree graphs with \(\Delta = 2^{O(\sqrt{\log n})}\), it does not imply that this runtime cannot be obtained for large-degree graphs with $\Delta=\poly(n)$, which would be sufficient to resolve Open Question~\ref{oq:main} in the affirmative. In fact, we conjecture that there is a \(o(\log n)\)-round algorithm for MIS in general graphs, and that the algorithm from Theorem~\ref{thm:main-girth} will likely play an important role in such an algorithm.

Moreover, by combining our algorithm with the low-arboricity to low-degree reduction by~\cite{BarenboimEPS16}, we break another barrier in trees. 

\begin{restatable}{thm}{thmtrees}\label{thm:trees}
    There is a randomized distributed algorithm that, with high probability, finds an MIS in  $O\left(\sqrt{\frac{\log n}{\log(\log^*n)}}\right)$ rounds in trees.
\end{restatable}

Theorem~\ref{thm:trees} improves upon the previous \( O(\sqrt{\log n}) \)-round algorithm implied by~\cite{Ghaffari16}, and it breaks a conjecture presented as part of Open Question 11.15 in the book by Barenboim and Elkin~\cite{2013Barenboim}, where they state:\footnote{The case of oriented trees is very different, as one can use the orientation to devise an $O(\log^* n)$-round algorithm~\cite{ColeV86}. However, one cannot hope to orient a tree efficiently, not even in $o(\log n)$ rounds~\cite{2013Barenboim}.}

\begin{center}
\begin{minipage}{0.85\textwidth}
\emph{``Pin down the randomized complexity of the MIS problem in unoriented trees. We conjecture that it is $\Theta(\sqrt{\log n})$.''}
\end{minipage}
\end{center}

%\seri{this thing down here needs more work. mainly clarify that the only known lower bound back then on trees was $\Omega(\log^* n)$. I don't think we should say anything honestly.... not even a rootlogn/loglogn was known for trees. let me think. We need to start first with state-of-the-art of MIS at the time of the publication. Only then say something about matching I think..}

%\aaron{Since this book was published, the context around this conjecture has changed, but the conjecture has remained open until our work.} 

At the time of the book’s publication, the best upper bound for trees was \(O(\sqrt{\log n \cdot \log \log n})\)~\cite{LenzenW11,BarenboimEPS16}, and the best lower bound was \(\Omega(\log^* n)\)~\cite{Linial92,Naor91}. As we discussed earlier, the upper bound was later improved to \(O(\sqrt{\log n})\) due to the algorithm by~\cite{Ghaffari16}, resolving the conjecture affirmatively from the upper bound side and leaving only the lower bound side open. Very recently, progress on MIS in trees seemed to be moving in the direction of resolving this conjecture affirmatively, due to an \(\Omega(\sqrt{\log n / \log \log n})\) lower bound by~\cite{BalliuGKO23}.

%\aaron{The formulation of this conjecture was likely influenced by the status of the closely-related maximal matching (MM) problem at the time.}
\paragraph{Separating MIS from Maximal Matching in Trees:} As discussed earlier, a very recent concurrent work~\cite{LBMatching} shows an $\Omega\left(\min\left\{\log\Delta,\sqrt{\log n}\right\}\right)$ lower bound for Maximal Matching (MM) in trees. Hence,  Theorem~\ref{thm:trees} implies a counterintuitive separation between MIS and MM in trees. MIS and MM are closely related, as one can find an MM by finding an MIS in the line graph.\footnote{In the line graph, the vertex set represents the set of edges of the original graph, and there is an edge between two vertices if the corresponding two edges intersect.} However, this reduction does not preserve the tree structure when applied to trees (e.g., if there is a node with degree at least three in the original tree, it would imply a triangle in the line graph). Theorem~\ref{thm:trees} shows that this is not just an artifact of the reduction, but a fundamental difference between the two problems in trees. That is, while MM can only be easier than MIS in general graphs, it is strictly \emph{harder} in trees. 
\paragraph{A Remark on Prior Work:} It is worth mentioning that the community had previously believed that an $\Omega\left(\min\left\{\log\Delta,\sqrt{\log n}\right\}\right)$ lower bound against randomized algorithms for MM in trees had been established 15 years ago via an extension of the lower bound of~\cite{KuhnMW16} (see also~\cite{abs-1011-5470}). However, as later clarified in~\cite{Bar-YehudaCS17} and in the full version of~\cite{KuhnMW16}, this belief was incorrect. The best-known lower bound for MM in trees (as well as in general graphs) until the very recent concurrent work by~\cite{LBMatching} remained only \(\Omega\left(\min\left\{\sqrt{\frac{\log n}{\log \log n}}, \frac{\log \Delta}{\log \log \Delta}\right\}\right)\)~\cite{KuhnMW16,BarenboimEPS16}. The recent work by~\cite{LBMatching} restores the lower bound to $\Omega\left(\min\left\{\log\Delta,\sqrt{\log n}\right\}\right)$ by employing a different technique than that used in~\cite{KuhnMW16}.
 %While this correction did not change the status of MIS in trees, it did open up an new question -- what is the complexity of MM in trees? The current gap between the upper and lower bounds is $\Theta(\sqrt{\log\log n})$.

%Even though there is a reduction from maximal matching to MIS in general graphs (by running an MIS algorithm on the line graph), this reduction does not hold in trees (as the line graph of a tree is not itself a tree). That is, while an algorithm for MIS implies an algorithm for maximal matching in general graphs, this is not the case in trees. In particular, our algorithm for trees does not apply to maximal matching, leaving the two possibilities of either a \(o(\sqrt{\log n})\)-round algorithm or an \(\Omega(\sqrt{\log n})\) lower bound for maximal matching in trees still both open.

%To get below the $\sqrt{\log n}$ barrier for MIS in trees, we needed to break the degree reduction and the constant survival probability barriers. To address MM in trees (and in general graphs), one would need to break these barriers for MIS in graphs that can contain very short cycles. Indeed, the line graph of a tree could contain, for instance,  many triangles.

\section{Technical Overview}\label{sec:TO}

The main idea behind the proof of Theorem~\ref{thm:main-girth} is a novel two-round algorithm with \(o(1)\) survival probability. To explain our algorithm, let us start with Luby's classical algorithm.

\paragraph{Luby's Algorithm~\cite{Luby86,AlonBI86}:} In this algorithm, in each round, each node picks a uniformly random rank \(r_v \in [0,1]\). If the rank of a node \(v\) is smaller than all of its neighbors' ranks \(r_u\), then \(v\) joins the independent set \(\mathcal{I}\). As we discussed earlier, in~\cite{Luby86,AlonBI86}, it is shown that removing \(\mathcal{I}\) together with the neighboring nodes deletes a constant fraction of the edges. Thus, after repeating this for \(O(\log n)\) rounds, the algorithm terminates with high probability.

\paragraph{Survival Probability:} It is not difficult to see that in regular graphs, Luby's algorithm deletes each node with constant probability after one round. Intuitively speaking, this stems from the \emph{balance} between the joining probability (i.e., the probability that a node joins \(\mathcal{I}\)) and the number of neighbors a node has. Specifically, in a \(\Delta\)-regular graph, each node joins with probability \(1/(\Delta+1)\), as this is the probability that a node's rank is the smallest among \(\Delta+1\) ranks. To show that a node is deleted with constant probability, we call a neighbor \(w\) of \(v\) a witness (for \(v\)'s deletion) if \(w\) has the smallest rank both in its neighborhood and in \(v\)'s neighborhood. For a neighbor \(w\) to be a witness, its rank must be the smallest among \(2\Delta - 1\) ranks, which occurs with probability \(1/(2\Delta - 1)\). Moreover, since \(v\) can have at most one witness, this implies that \(v\) is deleted with probability at least \(\Delta/(2\Delta - 1)\). Thus, \(v\) survives with probability \(\approx 1/2\).

If we could extend this survival probability guarantee beyond regular graphs (and beyond a single round), it would imply that Luby's algorithm itself shatters the graph into sufficiently small connected components after \(O(\log\Delta)\) rounds, after which we can complete the MIS in these components in \(\poly(\log\log n)\) rounds.\footnote{We elaborate on the shattering phase in the technical part in Section~\ref{sec:full}.} However, if the graph is not regular, then an \emph{unbalanced} node (i.e., a node \( u \) whose degree is \( \deg_u \) while the degree of any neighbor \( w \) is \( \deg_w \gg \deg_u \))
 survives with a probability \( > 1-o(1) \).

\paragraph{Beyond Regular Graphs via Ghaffari's Algorithm~\cite{Ghaffari16}:} The main intuition behind Ghaffari's algorithm is to balance the joining probabilities of the nodes over rounds, so that in most rounds, the behavior of the algorithm resembles Luby's on regular graphs. More specifically, each node \( v \) is associated with a parameter \( p_v \in [0,1] \) that captures its \emph{desire} to join the independent set in a given round. Initially, the values \(\{p_{v}\}_{v\in G}\) are all set to \(1/2\). Then, each node marks itself with probability \( p_v \) and joins \(\mathcal{I}\) if it is marked and none of its neighbors is marked.\footnote{This also resembles another version of Luby's algorithm, which was also given in the algorithm by Alon-Babai-Itai~\cite{AlonBI86}.} After each round, each node adjusts its desire probability based on the desire of its neighbors. Specifically, let \( d_v := \sum_{u \sim v} p_u \) be the \emph{total desire} in \( v \)'s neighborhood. If \( d_v \geq 2 \), then \( v \) decreases its desire by setting \( p_v = p_v / 2 \) for the next round. Otherwise, if \( d_v < 2 \), \( v \) increases its desire by setting \( p_v = \min\{2p_v, 1/2\} \) for the next round.

In~\cite{Ghaffari16}, it is shown that after \( O(\log\Delta) \) rounds, each node must have spent most of the rounds where it is deleted with constant probability. This allows the graph to shatter into sufficiently small connected components, resulting in the $O(\log\Delta+\polylog\log n)$ state-of-the-art running time for MIS.

To achieve a better running time, as a starting point, if we could attain a sub-constant survival probability per round in regular graphs, we could hope to use similar ideas as in~\cite{Ghaffari16} to maintain this guarantee for most rounds in general graphs. As an even simpler starting point, we could focus on regular graphs without any constant-length cycles. 

However, as we show in Theorem~\ref{thm:lb} (in Appendix~\ref{sec:lb}), no randomized one-round algorithm without unique identifiers can achieve better than a constant survival probability, even in regular graphs. Note that all randomized MIS algorithms in the literature, including those discussed above as well as ours, do not rely on vertex identifiers (unlike deterministic algorithms, which must have them). In fact, if unique identifiers are required for a randomized algorithm, they could be generated using the algorithm’s randomness.
Theorem~\ref{thm:lb} further highlights the limitations of prior techniques, as they repeatedly find an independent set in the remaining graph by using a one-round algorithm each time. To our surprise, Theorem~\ref{thm:lb} does not extend to \emph{two}-round algorithms, paving the way for our novel two-round algorithm with sub-constant survival probability.

\paragraph{Our Algorithm:} To explain our algorithm, let us revisit Luby's. Recall that in Luby's algorithm, a node \( v \) ``gives up" on joining the independent set \( \mathcal{I} \) in a round if it has a neighbor \( u \) with rank \( r_u < r_v \). In this case, we say that \( u \) \emph{dominates} \( v \). To achieve a fast running time, we would like a node \( v \) to reach its final status (i.e., whether it is in the computed MIS or not) as quickly as possible. Thus, if \( v \) is dominated by \( u \), and \( u \) ends up joining the independent set in that round, it helps \( v \) reach its final status immediately after that round.

Unfortunately, in several scenarios, \( v \) is dominated by a neighbor \( u \) that does not end up joining the independent set in that round. This happens when \( u \) itself has a dominating neighbor \( w \neq v \). In that case, \( v \) ``counted" on \( u \) to help it reach its final status by joining the independent set, but \( u \) fails to do so because it has a dominating neighbor of its own.

This opens up a potential avenue for improving Luby's Algorithm -- a vertex \( v \) should not \emph{fully} count on a vertex \( u \) that dominates it. In particular, when \(r_u = 0.49\) and \(r_v = 0.5\), the difference between the ranks of \(u\) and \(v\) is inconsequential and \(v\) should not give up all hope of joining the independent set. This raises an important question -- how much should \( v \) count on a dominating neighbor \( u \) to join the independent set? Naturally, the smaller \( r_u \) is, the more likely \( u \) is to join the independent set.

%When considering this algorithm, it might seem unfair that a node \( v \) gives up only if there is a neighbor with a smaller rank, particularly if that rank is not \emph{substantially smaller} than \( v \)'s rank. Indeed, if \( v \)'s rank is \( 0.5 \) and \( u \)'s rank is \( 0.49 \), then intuitively, \( u \) has a similar probability of being dominated itself by another neighbor \( w \), say with rank \( 0.48 \). Following this intuition, a node \( v \) should give up on joining the independent set in a given round only if it has a neighbor \( u \) with a rank substantially smaller than \( v \)'s. The smaller a neighbor \( u \)'s rank, the more likely it is for \( u \) to end up in the independent set.

This is precisely the main idea behind our two-round algorithm. Our algorithm begins like Luby's by assigning a uniformly random rank \( r_v \in [0,1] \) to each node \( v \). Then, a node \( v \) is called a candidate if its rank is \(\approx e^i / \Delta \) for some \( i \in [1, \dots, \ln \Delta] \), and none of its neighbors \( u \) has substantially smaller rank \( r_u \leq i / \Delta \). A node then joins the independent set if it is a candidate and none of its neighbors is a candidate. Interestingly, by not giving up on joining the independent set unless there is a neighbor with substantially smaller rank, our algorithm achieves a higher inclusion probability compared to Luby's, provided the graph has no triangles (by inclusion probability, we mean the probability that a node is included in the independent set). Specifically, the inclusion probability of our two-round algorithm in triangle-free graphs is \( \approx\log^* \Delta / \Delta \). Thus, in a neighborhood of a node, the expected number of nodes that join the independent set is \( \approx\log^* \Delta\), in contrast to Luby's algorithm, in which at most one neighbor joins the independent set in expectation.

The fact that a super-constant number of neighbors join the independent set in expectation is the crucial property that helps us achieve a sub-constant survival probability. We show that if there are also no cycles of length \(\leq 6\), then at least one neighbor joins the independent set with probability at least \( 1 - e^{\poly(\log^* \Delta)} \). Thus, a node reaches its final status (and is therefore deleted from the graph) with $1-o(1)$ probability at the end of our two-round algorithm, which implies a survival probability of $o(1)$.

We conclude our technical overview in Section~\ref{sec:warmup}, where we provide a formal proof for the higher inclusion probability guarantee of our algorithm in the simpler case of regular triangle-free graphs. Before that, we briefly outline the plan for extending this approach beyond regular graphs.

\paragraph{Combining Our Framework with Ghaffari's:} Now that we have a two-round algorithm with \(o(1)\) survival probability in regular graphs, we could hope to use the ideas in~\cite{Ghaffari16} to extend it to general graphs. However, the framework of~\cite{Ghaffari16} encounters some complications in our setting, which we overcome as follows.

As a first step, we need to devise a two-round algorithm that takes into account the desire weights \(\{p_v\}_{v\in G}\). Roughly speaking, since we want to achieve a higher inclusion probability (which would imply a lower survival probability), we need each node \( v \) to join the independent set with a probability higher than its desire probability \( p_v \). For this, we use an idea conceptually similar to the regular graph case (though differing in technical details) to devise a two-round algorithm with the following guarantee: if the total desire \( d_v = \sum_{u \sim v} p_u \) in the neighborhood of \( v \) is less than \(\approx (\log^{*}\Delta)^{1/100} \), then \( v \) joins the independent set with probability \( p_v \cdot \poly(\log^*\Delta) \). To ensure this does not exceed 1, our algorithm caps the probability at a sufficiently small value \( \ll \poly(1/\log^{*}\Delta) \), unlike in~\cite{Ghaffari16}, where probabilities are capped at \( 1/2 \).

Furthermore, to achieve a better overall running time, we need to increase or decrease the desire probabilities by a factor higher than 2. Instead, we adjust the desire probabilities by a factor of \( \ell := \poly(\log^*\Delta) \) after each call of the two-round algorithm. Specifically, if the total neighborhood desire \( d_v = \sum_{u \sim v} p_u \geq \ell \), then \( v \) decreases its desire by setting \( p_v = p_v / \ell \). Otherwise, if \( d_v < \ell \), then \( v \) increases its desire by updating \( p_v = \min\{\ell p_v, \tau\} \), where \( \tau \ll \poly(1/\log^*\Delta) \).

\paragraph{Our Post-Shattering Phase:} Interestingly, our post-shattering phase differs slightly from prior work. Instead of showing that the graph shatters into sufficiently small connected components after \( O(\log_{\ell}\Delta) = O(\log\Delta / \log(\log^*\Delta)) \) rounds, we show that it shatters into two parts, \( H \) and \( H' \). In \( H \), each connected component is sufficiently small, allowing us to complete the MIS in this part using standard techniques from prior work. In \( H' \), the maximum degree is only \( \Delta'=O(\ell / \tau) \), allowing us to use the deterministic \( O(\Delta' + \log^* n) \)-round algorithm in \( H' \) as a black-box to finish the MIS in this part.

\subsection{The Core Idea: Higher Inclusion Probability
}\label{sec:warmup}

%\seri{Double check that independent set is $\mathcal{I}$ and not $I$.}\seri{double check that intervals are $I$ and not $B$}\seri{Check that neighbors in bucket $i$ are $N_i(v)$ and not $B$}
In this section we show that there exists a two-round algorithm that includes each node with probability higher than $1/\Delta$ in $\Delta$-regular triangle-free graphs. 

\begin{lem}\label{thm:log-star-lower}
In $\Delta$-regular triangle-free graphs, there is a two-round algorithm that finds an independent set $\mathcal{I}$  such that the probability that a node $v$ is included in $\mathcal{I}$ is at least $\frac{\log^* \Delta}{10^4\Delta}$.
\end{lem}

    Before providing the algorithm, we set up some notation. We denote the input graph by $G$, the set of nodes by $V(G)$, and the set of edges by $E(G)$. We assume without loss of generality that $\Delta > 1000$.\footnote{Otherwise, Luby's one-round algorithm includes each node with probability $1/(\Delta+1) > \log^* \Delta / (1000\Delta)$ for $\Delta \leq 1000$.} Let $k = \lfloor (\log^* \Delta)/10 \rfloor$, $a_1 = 5$, and $a_{i+1} = e^{a_i - 3}$ for every integer $i \in [2, \cdots, k]$. Furthermore, let $b_i = a_i/\Delta$, and define the intervals $I_i = (b_i, b_{i+1}]$ and $J_i = (b_i, 1]$ for each integer $i \in [k]$. Observe that \( a_i < \Delta^{1/10} \), and that all intervals \( I_i \) and \( J_i \) are contained within the interval \( [0, 1] \), since \( \frac{5}{\Delta} \leq b_i \leq \frac{1}{100} \) by the definition of \( k \). Finally, we use the standard notation \( u \sim v \) to refer to a neighbor \( u \) of \( v \).

%For any $i\in [k]$, let $\mathcal{S}_i \subseteq \mathcal{S}$ denote the family of multisets of $\Delta-1$ elements all taken from $J_i$. $a_{i+1} \le \frac{1}{e^{-2(2a_i/\Delta)(\Delta-2)}} \le e^{4a_i}$. Thus, by definition of $k$, $b_i\le 1/100$ for all $i\in [k+1]$. Furthermore, $a_{i+1}\ge e^{2a_i(\Delta-2)/\Delta} \ge 10 a_i$ and $b_{i+1}\ge 10 b_i$ for all $i\in [k]$ since $a_1 \ge 5$.\seri{do we need the 2-factor in $b_i$?}

The algorithm for proving Lemma~\ref{thm:log-star-lower} is provided in Algorithm~\ref{alg:2-round}. Each node $v$ picks a uniformly random rank $r_v \in [0,1]$. If $r_v\in I_i$ for some $i\in [k]$ and $r_u\in J_i$ for all $u\sim v$, then $v$ joins a set of candidates $C$.\footnote{Intuitively speaking, $v$ is a candidate if its rank falls into some interval $I_i$ and none of the neighbors' ranks falls into a smaller interval $I_j$ for $j<i$.} If $v$ is a candidate and none of its neighbors is a candidate, then $v$ is added to the independent set $\mathcal{I}$. The algorithm can be implemented in two rounds: in the first round, each node picks a random rank and sends it to its neighbors. In the second round, each node informs its neighbors whether it is a candidate and decides whether to join the independent set accordingly.%\seri{a reminder for myself to not repeat things again after technical overview..}

\begin{algorithm}
\caption{Two-round algorithm for proving Lemma~\ref{thm:log-star-lower}}\label{alg:2-round}
\begin{algorithmic}[1]
\State Initialize the independent set $I \gets \emptyset$ and the candidate set $C \gets \emptyset$
\For{each node $v \in V(G)$}
    \State Sample a uniformly random rank $r_v \in [0,1]$ independently
    \If{$r_v \in I_i$ for some $i \in [k]$}
        \State Let $i(v)$ be the unique value for which $r_v \in I_{i(v)}$
        \If{$r_u \in J_{i(v)}$ for all $u \sim v$}
            \State Add $v$ to the set of candidates $C$
        \EndIf
    \EndIf
\EndFor
\For{each node $v \in C$}
    \If{$u \notin C$ for all $u \sim v$}
        \State Add $v$ to the independent set $\mathcal{I}$
    \EndIf
\EndFor
\end{algorithmic}
\end{algorithm}

\begin{proof}[Proof of Lemma~\ref{thm:log-star-lower}]
Clearly, Algorithm~\ref{alg:2-round} produces an independent set $\mathcal{I}$ as no candidate joins $\mathcal{I}$ if it has a neighboring candidate. Hence, we focus on the probability that $u$ is added to the independent set (i.e., the inclusion probability). %\seri{just fix the ``ands below"}

~\\\noindent\textbf{Inclusion probability}: We fix the node $v$ for the rest of the proof. Let $C_i$ be the set of candidates in $C$ with ranks in $I_i$, i.e., $C_i = \{u \in C \mid r_u \in I_i\}$. First, we show that the probability that $v \in C_i$ is $\approx 1/\Delta$ for any $i \in [k]$. Then, we show that the probability that $v \in \mathcal{I}$ given that $v \in C_i$ is at least a constant. Hence, since the $I_i$'s are disjoint, summing over all $i$'s gives an inclusion probability of $\approx \log^* \Delta / \Delta$.

~\\\noindent\textbf{Lower bounding $\Pr[v\in C_i]$:} Observe that:

\begin{align}
    \Pr[v\in C_i]  &= \Pr [ r_v\in I_i \text{ and } (r_u\in J_i \text{ } \forall u\sim v)]\\ &= (b_{i+1} - b_i)\cdot (1-b_i)^\Delta\\
    &\geq \frac{a_{i+1} - a_i}{\Delta}\cdot \exp\left(-\Delta( b_i + 2b_i^2)\right)\\
    &\geq \frac{a_{i+1}}{2\Delta}\cdot \exp\left(-a_i - 2a_i^2/\Delta\right)\\
    &\geq \frac{a_{i+1}}{2\Delta}\cdot \exp\left(-a_i - 1/2\right)\\
    &\geq \frac{a_{i+1}}{2\Delta}\cdot \frac{1}{e^4\cdot a_{i+1}}\\
    &=\frac{1}{2\cdot e^4\cdot \Delta}
\end{align}

where (2) follows by independence, (3) follows by a standard argument for approximating exponents (see also Proposition~\ref{prop:exp-approx}), (5) follows since $a_i< \Delta^{1/10}$ and $\Delta>1000$, which implies that $2a_i^2/\Delta < 1/2$, and (6) follows sins $a_{i+1}\cdot e^4 = e^{a_i-3}\cdot e^4 > e^{a_i}\cdot \sqrt{e} = \exp(a_i+1/2)$.

~\\\noindent\textbf{Lower bounding $\Pr[v\in \mathcal{I}\mid v\in C_i]$:} The main idea is to notice that if $v \in C_i$, then whether it is included in the independent set $\mathcal{I}$ depends solely on its neighbors $u$ with ranks $r_u \in I_i$. This is because if $v \in C_i$, then all of its neighbors $u$ have ranks $r_u \in J_i$. Furthermore, those with ranks $r_u \in I_j$ for $j > i$ cannot be candidates (due to Line 6 in Algorithm~\ref{alg:2-round}) because they have a neighbor, $v$ itself, with rank $r_v \notin J_j$. Hence, to bound the probability that a node $v \in C_i$ is added to the independent set, it suffices to bound the probability that none of the neighbors $u$ with $r_u \in I_i$ are in $C$ (the other neighbors are not in $C$ with probability 1).

For this, we first upper bound the number of neighbors of $v\in C_i$ with ranks in  $I_i$, then we upper bound the probability that a given such neighbor is a candidate, and finally we use a standard union bound argument to show that none of these neighbors is a candidate with constant probability. 

Let $N_i(v)$ be the set of neighbors $u$ of $v$ with $r_u\in I_i$, and let $s_i = |N_i(v)|$. Observe that for a neighbor $u$ of $v$, $\Pr[u\in N_i(v)\mid v\in C_i] = (b_{i+1}-b_i)/(1-b_i)$ (since $v\in C_i$, we know that $u$ is not in the interval $[0,b_i]$). Hence, we get that $\mathbb{E}[s_i\mid v\in C_i] \leq  \Delta\cdot b_{i+1}/(99/100)\leq 2\Delta\cdot b_{i+1} = 2a_{i+1}$. Therefore, by Markov's inequality, $\Pr[s_i < 4a_{i+1}\mid v\in C_i]\geq 1/2$. 

Furthermore, for a neighbor $u\in N_i(v)$ it holds that $\Pr[u\in C\mid v\in C_i] =(1-b_i)^{(\Delta-1)}$. This is because for such a neighbor $u$, since $r_u\in I_i$, the probability that it is in $C$ is exactly the probability that it does not have a neighbor $w\neq v$ with $r_w\notin J_i$. Since the graph is triangle-free, $v$ and $u$ have no common neighbors. Therefore, the event that $u$ has no neighbor $w\neq v$ with $r_w\in J_i$ is independent of the event that $v\in C_i$. Hence, for a neighbor $u$ of $v$, $\Pr[r_w\in J_i \text{ } \forall w\sim u \mid v\in C_i] =(1-b_i)^{(\Delta-1)}\leq e^{-a_i(\Delta-1)/\Delta} = e^{-a_i + a_i/\Delta}\leq e^{-a_i + 1/2}$. By a union bound, if $s_i<4a_{i+1}$, we get that $\Pr[u\notin C \text{ } \forall u\in N_i(v) \mid v\in C_i]\geq 1-4a_{i+1}\cdot e^{-a_i + 1/2} \geq 1- 4/e^2$. %\aaron{minor issue: $s_i$ is a random variable. Acknowledge that we need to fix all randomness on $v$ and all neighbors of $v$ in order to make it deterministic. I.e. the conditioning should not be over $v\in C_i$, but on some specific fixing for the 1-neighborhood of $v$.}\seri{let me know if it makes sense now}

Now we are ready to give a lower bound on $\Pr[v\in \mathcal{I}\mid v\in C_i]$. Observe that:
\setcounter{equation}{0}

\begin{align}
    \Pr[v\in \mathcal{I}\mid v\in C_i] &= \Pr[u\notin C \text{ } \forall u\sim v \mid v\in C_i]\\ &= \Pr[u\notin C \text{ } \forall u\in N_i(v) \mid v\in C_i]\\
    &\geq \Pr[u\notin C \text{ } \forall u\in N_i(v)\mid v\in C_i, s_i < 4a_{i+1})]\cdot \Pr[s_i < 4a_{i+1}\mid v\in C_i]\\
    &\geq \left(1-4a_{i+1}\cdot \frac{1}{e^2\cdot a_{i+1}}\right)\cdot (1/2)\\
    &\geq \left(1-4e^{-2}\right)\cdot (1/2)\\
    & \geq \frac{1}{5}
\end{align}

where (2) and (4) follow from the discussion above. Given the lower bounds on $\Pr[v\in C_i]$ and $\Pr[v\in \mathcal{I}\mid v\in C_i]$, we are ready to give a lower bound on the probability that a node $v$ is included in $\mathcal{I}$. 

~\\\noindent\textbf{Lower bounding $\Pr[v\in \mathcal{I}]$:} Since the $I_i$'s are disjoint, we have that:

\begin{align*}
\Pr[v\in \mathcal{I}] &= \sum_{i=1}^k \Pr[v\in C_i \cap \mathcal{I}]\\
&=\sum_{i=1}^k \Pr[v\in \mathcal{I}\mid v\in C_i]\cdot \Pr[v\in C_i]\\
&\geq \sum_{i=1}^k \frac{1}{5}\cdot \frac{1}{2\cdot e^4\Delta}\\
&\geq \frac{\log^*\Delta}{10^4\Delta}
\end{align*}

as desired.
%and each of the neighbors $v$ candidate and none of it neighbors is a candidate. Hence, $u$ joins $\mathcal{I}$ if and only if if there exists an $i$ for which $r_u\in I_i$, for any neighbor $v$ of $u$ $r_v\in J_i$, and all $j$s for which $y_j\in I_i\cap J_i$ have the property that $S_j\notin \mathcal{S}_i$. The probability of this happening is exactly the following:

%$$\Pr_{z\sim \mathcal{R}_2}[f(z) = 1] = \sum_{i=1}^k (b_{i+1} - b_i) \sum_{\ell = 0}^{\Delta} \binom{\Delta}{\ell} (b_{i+1} - 2b_i)^{\ell}(1 - b_{i+1})^{\Delta-\ell}(1 - (1 - 2b_i)^{\Delta-1})^{\ell}$$

%Thus,

%\begin{align*}
%\Pr_{z\sim \mathcal{R}_2}[f(z) = 1] &= \sum_{i=1}^k (b_{i+1} - b_i) \sum_{\ell = 0}^{\Delta} \binom{\Delta}{\ell} (b_{i+1} - 2b_i)^{\ell}(1 - b_{i+1})^{\Delta-\ell}(1 - (1 - 2b_i)^{\Delta-1})^{\ell}\\
%= \sum_{i=1}^k (b_{i+1} - b_i)\left((b_{i+1} - 2b_i)(1 - (1 - 2b_i)^{\Delta-1}) + (1 - b_{i+1})\right)^{\Delta}\\
%&= \sum_{i=1}^k (b_{i+1} - b_i)(1 - 2b_i)^2(1 - 2b_i)^{\Delta-2}\left(1 - (b_{i+1} - 2b_i)(1 - 2b_i)^{\Delta-2})\right)^{\Delta}\\
%&= \sum_{i=1}^k (b_{i+1} - b_i)(1 - 2b_i)^2\frac{1}{a_{i+1}}\left(1 - \frac{b_{i+1} - 2b_i}{a_{i+1}}\right)^{\Delta}\\
%&\ge \sum_{i=1}^k \frac{9}{10\Delta}\left(\frac{49}{50}\right)^2(1 - 1/\Delta)^{\Delta}\\
%&\ge \frac{\log^* \Delta}{100\Delta}
%\end{align*}

%as desired.
\end{proof}

\section{Faster Degree Reduction}\label{sec:mainalg}

In this section we prove Theorem~\ref{thm:main-girth}, which we restate here for convenience. At the end of the section we also provide a proof for Theorem~\ref{thm:trees}.

\thmmain*

To prove Theorem~\ref{thm:main-girth}, we first present a generalization to Lemma~\ref{thm:log-star-lower} beyond regular graphs in Section~\ref{sec:weightedInclusion}. Then, in Section~\ref{sec:weightedSurvival}, we derive an upper bound on the survival probability in graphs with girth at least $7$. Finally, in Section~\ref{sec:full} we give the full algorithm. Interestingly, the only part of the proof that assumes that the graph has girth at least $7$ is the upper bound on the survival probability in Lemma~\ref{lem:inclusion-weighted-logstar}. For the other parts, it suffices that the graph is triangle-free.

\subsection{Higher Inclusion Probability Beyond Regular Graphs}\label{sec:weightedInclusion}

In this section, we extend Lemma~\ref{thm:log-star-lower} to apply beyond regular graphs. Each node \( v \) is associated with a weight \( p_v \in [0,1] \), which can be thought of as a measure of \( v \)’s ``desire" to join the independent set. In the regular case, the analogous desire for each node is \( 1/\Delta \), and Algorithm~\ref{alg:2-round} includes each node with probability \(\approx \log^*\Delta/\Delta \). Here, we present an algorithm that includes a node \( v \) with probability \( \poly(\log^*\Delta) \cdot p_v \), as long as the total desire \( \sum_{u\sim v} p_u \) in its neighborhood is small. The formal statement is given in the following lemma.

\begin{lem}\label{lem:inclusion-weighted-logstar} Let $\tau := 1/(\log\log\log \Delta)$, $\ell:= (\log^*(1/\tau))^{1/100}$, and $\beta := \log^*(1/\tau)/\ell$.
Consider a triangle-free graph $G$ with node weights $\{p_v\}_{v\in V(G)}$ with $0 \le p_v \le \tau$. For any $v\in V(G)$, let $d_v := \sum_{u\sim v} p_u$. Then, there is a randomized algorithm that finds an independent set $\mathcal{I}$ such that each node $v\in V(G)$ with $d_v \le \ell$ is included in $\mathcal{I}$ with probability at least $p_v \beta/10^6$.
\end{lem}

\begin{table}[ht]
\renewcommand{\arraystretch}{1.5}
\centering
\begin{tabular}{|c|l|l|}
\hline
\textbf{Notation} & \textbf{Value} & \textbf{Note} \\ \hline
$G$ & & The input graph \\ \hline
$V(G)$ & & The set of nodes \\\hline
$E(G)$ & & The set of edges \\\hline
$\Delta$ & & The maximum degree in the graph. \\ \hline
$\tau$ & $1/(\log\log\log \Delta)$ &  \\ \hline
$\ell$ & $(\log^*(1/\tau))^{1/100}$ & \\ \hline
$\beta$ & $\log^*(1/\tau)/\ell$ & \\
\hline
$p_v$ & Variable & \makecell[l]{A weight associated with node $v$; \\ intuitively, $p_v$ is the ``desire" of $v$ to join $\mathcal{I}$.} \\ \hline
$d_v$ & $\sum_{u \sim v} p_u$ & \makecell[l]{Total ``desire" in the neighborhood of $v$.} \\ \hline
$k$ & $\log^*(1/\tau)/100$ & The number of intervals. \\ \hline
$a_1$ & $10 \beta$ & Initial value in the $a_i$ sequence. \\ \hline
$a_{i+1}$ & $e^{\ell a_i}/(16\ell)$ & Recursive definition for $a_i$. \\ \hline
$b_{v,i}$ & $a_i p_v$ & \makecell[l]{A parameter used to adjust the boundary \\ of the intervals $I_{v,i}$ and $J_{u,v,i}$.} \\ \hline
$I_{v,i}$ & $(b_{v,i}, b_{v,i+1}]$ & \makecell[l]{Interval associated with node $v$ and index $i$.} \\ \hline
$J_{v,u,i}$ & $(\ell b_{u,i}/d_v, 1]$ & \makecell[l]{Interval associated with edge  $\{u,v\}$ and index $i$.} \\ \hline
\end{tabular}
\caption{Table of Notations}
\label{table}
\end{table}

Before we present the proof of Lemma~\ref{lem:inclusion-weighted-logstar}, we set up some notation. For convenience, a list of notations is provided in Table~\ref{table}. Let $k = \log^*(1/\tau)/100$, $a_1 := 10 \beta$ and, for every $i \in [k]$, let $a_{i+1} := e^{\ell a_i}/(16\ell)$. For every $i \in [k]$, let $b_{v,i} := a_i p_v$. For any $v \in V(G)$ and $i \in [k]$, define the interval $I_{v,i} := (b_{v,i}, b_{v,i+1}]$. For any $\{u,v\} \in E(G)$ and $i \in [k]$, define the interval $J_{v,u,i} := (\ell b_{u,i}/d_v, 1]$. 

The algorithm for proving Lemma~\ref{lem:inclusion-weighted-logstar} is given in Algorithm~\ref{alg:weightedInclusion}. Each node \( v \) samples a uniformly random rank \( r_v \in [0,1] \). Then, based on the ranks, we construct a set of candidates as follows. If \( d_v \leq \sqrt{\tau} \) and \( r_v \leq \beta p_v \), then \( v \) is added to a set of candidates \( C^- \). Otherwise, in the case that \( \sqrt{\tau}\leq d_v \leq \ell \), if \( r_v\in  I_{v,i} \) for some \( i \in [k] \) and \( r_u\in  J_{v,u,i} \) for any neighbor $u$ of $v$, then we add $v$ to a set of candidates $C^+_i$. Finally a node joins the independent set \( \mathcal{I} \) if it is a candidate in \( C = C^- \cup (\cup_{i=1}^k C_i^+) \) and none of its neighbors is a candidate. Before we prove the correctness of the algorithm, let us prove the following useful proposition.

\begin{algorithm}
\caption{Inclusion$(G, p)$}
\begin{algorithmic}[1]
\State Initialize the independent set $I\gets \emptyset$, and the sets of candidates $C, C^-, C_i^+ \gets \emptyset$ for all $i \in [k]$.
\State For every node $v\in V(G)$, sample a uniformly random rank $r_v \in [0,1]$ independently.

\For{each $v \in V(G)$}
    \If{$d_v \leq \sqrt{\tau}$}
        \State Add $v$ to $C^-$ if $r_v \leq \beta p_v$
    \ElsIf{$d_v \leq \ell$}
        \If{$r_v \in I_{v,i}$ for some $i \in [k]$}
            \State Let $i(v)$ be the unique value for which $r_v \in I_{v,i(v)}$
            \State Add $v$ to $C_{i(v)}^+$ if $r_u \in J_{v,u,i(v)}$ for all $u \sim v$
        \EndIf
    \EndIf
\EndFor

\State $C \gets C^- \cup (\cup_{i=1}^k C_i^+)$.
\For{each $v \in V(G)$}
    \If{$v \in C$ and $u \notin C$ for all $u \sim v$}
        \State Add $v$ to $\mathcal{I}$
    \EndIf
\EndFor
\State \Return $\mathcal{I}$
\end{algorithmic}\label{alg:weightedInclusion}
\end{algorithm}

\begin{prop}\label{prop:same-adj}
For any two neighbors $u,v\in C$, either $u,v\in C^-$, or $u,v\in C_i^+$ for some $i\in [k]$.
\end{prop}

\begin{proof}
Suppose that $u\in C_i^+$ for some $i\in [k]$. By definition of $C^+_i$, $d_u\leq \ell$ and $r_v\in J_{u,v,i} = (\ell b_{v,i}/d_u,1]$, which implies that $r_v > b_{v,i}$. Hence, $r_v\notin I_{v,j}$ for any $j < i$, which implies that $v\notin C_j^+$ for any $j < i$. Furthermore, since $r_v > b_{v,i} = a_i p_v \geq a_1 p_v > \beta p_v$, it holds that $v\notin C^-$. Thus, if $v\in C$, it must be the case that $v\in C_j^+$ for some $j\ge i$. Applying the same reasoning for $v$ shows that $i\ge j$, which implies that $i = j$ as desired. 

Furthermore, if $u\in C^-$ instead, the argument above shows that $u$ has no neighbor in $\bigcup_{i=1}^k C^+_i$, as desired.
\end{proof}

\begin{proof}[Proof of Lemma \ref{lem:inclusion-weighted-logstar}]
Clearly, Algorithm~\ref{alg:weightedInclusion} produces an independent set, as no candidate joins the independent if it has a neighboring candidate. Thus, we focus on lower bounding the probability that a vertex $v$ with $d_v\le \ell$ is in $\mathcal{I}$. First, consider the case where $d_v \le \sqrt{\tau}$. In this case,

\begin{align*}
\Pr[v\in \mathcal{I}] &= \Pr[v\in C^- \text{ and } u\notin C \text{ for all } u\sim v]\\
&= \Pr[v\in C^- \text{ and } u\notin C^- \text{ for all } u\sim v]\\
&= \Pr[u\notin C^- \text{ for all } u\sim v \mid v\in C^-]\Pr[v\in C^-]\\
&\ge \left(1 - \sum_{u\sim v} \beta p_u\right) (\beta p_v)\\
&\ge (1 - \beta \sqrt{\tau})\beta p_v\\
&\ge \beta p_v/100
\end{align*}

where the second equality follows from Proposition \ref{prop:same-adj}. This completes the first case. Next, we move onto the second case, i.e. when $\sqrt{\tau} < d_v \le \ell$. Similarly to the proof of Lemma~\ref{thm:log-star-lower}, we first lower bound $\Pr[v\in C^+_i]$, and then we lower bound $\Pr[v\in \mathcal{I}\mid v\in C^+_i]$. To lower bound the latter, we first upper bound the number of neighbors $u$ of $v$ with ranks in $I_{u,i}$ (these are the possible candidate neighbors if $v\in C^+_i$), then we analyze the probability that a given such neighbor is not a candidate, and finally we use a union bound argument to show that none of these neighbors is a candidate with constant probability. We start with $\Pr[v\in C_i^+]$. 

~\\\noindent \textbf{Lower bounding $\Pr[v\in C^+_i]$:} Notice that for any $i\in [k]$,

\setcounter{equation}{0}

\begin{align}
\Pr[v\in C_i^+] &= \Pr[r_v\in I_{v,i}]\prod_{u\sim v}\Pr[r_u\in J_{v,u,i}]\\
&\geq \frac{b_{v,i+1}}{2}\prod_{u\sim v}\left(1 - \frac{\ell b_{u,i}}{d_v}\right)\\
&\ge \frac{b_{v,i+1}}{2}\prod_{u\sim v}\exp\left(-\frac{\ell a_i p_u}{d_v} - 2\frac{(\ell a_i)^2p_u^2}{d_v^2}\right)\\
&> \frac{b_{v,i+1}}{2}\prod_{u\sim v}\exp\left(-\frac{\ell a_i p_u}{d_v} - 2\frac{\tau (\ell a_i)^2p_u}{\sqrt{\tau}d_v}\right)\\
&= \frac{b_{v,i+1}}{2} e^{-\ell a_i-2\sqrt{\tau}(\ell a_i)^2}\\
&> \frac{a_{i+1}}{6} e^{-\ell a_i}p_v\\
&\geq \frac{p_v}{96\ell}
\end{align}

where (3) follows from a standard argument for approximating exponents (see also Proposition \ref{prop:exp-approx}), (6) follows since $2\sqrt{\tau}(\ell a_i)^2\leq 1$, as $a_i<(1/\tau)^{1/10}$ for any $i\in [k]$ and $\ell<\log^*(1/\tau)$, and (7) follows since $a_{i+1} = e^{\ell a_i}/(16\ell)$.

~\\\noindent\textbf{Lower bounding $\Pr[v\in \mathcal{I}\mid v\in C^+_i$]:} First, we analyze the number of neighbors $u$ of $v$ with ranks $r_u\in I_{u,i}$, and then we show that none of them is a candidate with constant probability. Let $N_i(v)$ denote the set of neighbors $u$ of $v$ with  $r_u \in I_{u,i}$, and let $s_i(v) = |N_i(v)|$. Observe that:

\setcounter{equation}{0}

\begin{align}
\mathbf{E}[s_i(v) \mid v \in C_i^+] &= \sum_{u\sim v} \Pr[r_u\in I_{u,i} \mid v\in C_i^+]\\
&= \sum_{u\sim v} \Pr[r_u \in I_{u,i} \mid r_u \in J_{v,u,i}]\\
&\leq \sum_{u\sim v} \frac{b_{u,i+1}}{1 - \ell b_{u,i}/d_v}\\
&\leq \sum_{u\sim v} \frac{b_{u,i+1}}{1 - \ell a_i p_u/d_v}\\
&\le \sum_{u\sim v} \frac{a_{i+1}p_u}{1 - \ell a_i \sqrt{\tau}}\\
&\le \sum_{u\sim v} \frac{a_{i+1}p_u}{1/2}\\
&\le 2a_{i+1}d_v\\
&\le 2a_{i+1}\ell
\end{align}

where (5) follows since $p_u\leq \tau$ and $d_v\geq \sqrt{\tau}$, and (6) follows since $\ell a_i << 1/(2\sqrt{\tau})$ for any $i\in [k]$. Thus, by Markov's Inequality,

$$\Pr[s_i \le 4 a_{i+1}\ell \mid v\in C_i^+] \ge 1/2$$

Moreover, since the graph is triangle-free $u$ and $v$ have no common neighbors. Thus, for a neighbor $u\in N_i(v)$ of $v\in C^+_i$ we have that:

\setcounter{equation}{0}

\begin{align}
\Pr[u \in C \mid v\in C^+_i] &= \Pr[u \in C^+_i \mid v\in C^+_i]\\
&=\prod_{w\sim u: w\ne v} \Pr[r_w \in J_{u,w,i}]\\
&= \prod_{w\sim u: w\ne v} (1 - \ell b_{w,i}/d_u)\\
&\le \prod_{w\sim u: w\ne v} \exp(-\ell a_i p_w/d_u)\\
&= e^{-\ell a_i + \ell a_i p_v/d_u}\\
&\le e^{-\ell a_i + \ell a_i\sqrt{\tau}}\\
&\le 2e^{-\ell a_i}
\end{align}

where (1) follows by Proposition~\ref{prop:same-adj}, and (6) follows since $p_v\leq \tau$ and $d_u\geq \sqrt{\tau}$ (otherwise $u \notin N_i(v)$), and $\ell a_i << 1/(2\sqrt{\tau})$. Thus, by a union bound,

$$\Pr[\exists u\sim v: u\in C_i^+ \mid v\in C_i^+,s_i(v)\leq 4a_{i+1}\ell] \le 4e^{-\ell a_i} s_i(v) \le 8e^{-\ell a_i }a_{i+1}\ell\leq 1/2$$

where the last inequality follows since $a_{i+1} = e^{\ell a_i}/(16\ell)$. Now we are ready to lower bound $\Pr[v\in \mathcal{I}\mid v\in C^+_i]$. Combining the calculations above shows that:

\setcounter{equation}{0}

\begin{align}
    \Pr[v\in \mathcal{I}\mid v\in C_i^+] &= \Pr[\forall u~\sim v: u\notin C\mid v\in C^+_i] \\&= \Pr[\forall u~\sim v: u\notin C^+_i\mid v\in C^+_i]\\
    &\geq \Pr[\forall u~\sim v: u\notin C^+_i\mid v\in C^+_i, s_i(v)\leq 4a_{i+1}\ell] \Pr[s_i(v)\leq 4a_{i+1}\ell\mid v\in C^+_i]\\
    &\geq 1/4
\end{align}

where (2) follows from Proposition~\ref{prop:same-adj}, and (4) follows since $\Pr[\forall u~\sim v: u\notin C^+_i\mid v\in C^+_i, s_i(v)\leq 4a_{i+1}\ell] = 1- \Pr[\exists u\sim v: u\in C_i^+ \mid v\in C_i^+, s_i(v)\leq 4a_{i+1}\ell]\geq 1/2$ and $\Pr[s_i(v)\leq 4a_{i+1}\ell\mid v\in C^+_i]\geq 1/2$, as discussed above. Finally, by Proposition \ref{prop:same-adj}, and since the intervals $I_{v,i}$ are disjoint for any $v$, we have that when $\sqrt{\tau}\le d_v \le \ell$,

$$\Pr[v\in \mathcal{I}] = \sum_{i=1}^k \Pr[v\in \mathcal{I}\mid v\in C_i^+]\Pr[v\in C_i^+]\ge \frac{kp_v}{384\ell} \ge \beta p_v/10^6$$

completing the second case, as desired.
\end{proof}

A small modification to the proof of Lemma \ref{lem:inclusion-weighted-logstar} proves the following:

\begin{lem}\label{lem:better-inclusion-weighted-logstar}
Let $G$ be a triangle-free graph. For any $\{u,v\}\in E(G)$ with $d_v \le \ell$ and node weights $p$ with $0\le p_v\le \tau$ for all $v\in V(G)$, Algorithm~\ref{alg:weightedInclusion} has the property that:

$$\Pr[v\in \mathcal{I}\mid r_u \ge \tau^{1/10}] \ge p_v\beta/10^6$$
\end{lem}

where $\ell, \tau, \beta$, and $d_v$ are as defined in Lemma~\ref{alg:weightedInclusion}.
\begin{proof}
The proof is very similar to the proof of Lemma~\ref{lem:inclusion-weighted-logstar}. Observe that if $r_u>\tau^{1/10}$, then $r_u\notin I_{u,i}$ and $r_u\in J_{v,u,i}$ for any $i\in [k]$. Therefore, $u\notin C=C^-\cup (\bigcup_{i=1}^k C_i^+)$ with probability $1$, given that $r_u>\tau^{1/10}$. This implies that the lower bounds derived on $\Pr[v\in \mathcal{I}]$ in the proof of Lemma~\ref{lem:inclusion-weighted-logstar} apply to $\Pr[v\in \mathcal{I}\mid r_u\geq \tau^{1/10}]$. In particular, going over the same case analysis an in the proof of Lemma~\ref{lem:inclusion-weighted-logstar}, if $d_v\leq \sqrt{\tau}$, we have that:

\setcounter{equation}{0}
    \begin{align}\Pr[v\in \mathcal{I}\mid u\notin C] &= \Pr[v\in C^- \text{ and } u'\notin C \text{ for all } u'\sim v\mid u \notin C]\\
    &\geq\Pr[v\in C^- \text{ and } u'\notin C \text{ for all } u'\sim v]\\
    &\geq \beta p_v/100
    \end{align}
where (2) follows since $u$ is a neighbor of $v$ and (3) follows from the calculation in the proof of Lemma~\ref{lem:inclusion-weighted-logstar} for the case where $d_v\leq \sqrt{\tau}$. For the other case where $\sqrt{\tau}\leq d_v\leq \ell$, observe that $\Pr[v\in C^+_i\mid r_u\in J_{v,u,i}]\geq \Pr[v\in C^+_i]\geq p_v/(96\ell)$, where the last inequality follows from the calculation in the proof of Lemma~\ref{lem:inclusion-weighted-logstar}. Similarly, $\Pr[v\in \mathcal{I}\mid v\in C^+_i, u\notin C]\geq \Pr[v\in \mathcal{I}\mid v\in C^+_i]\geq 1/4$, which implies that $\Pr[v\in \mathcal{I}\mid r_u\geq \tau^{1/10}]\geq \beta p_v/10^6$, as desired. 
\end{proof}

\subsection{Deriving A Small Survival Probability}\label{sec:weightedSurvival}

%\seri{Upper bounding survival probability. The key nice observation here is that in the case in which the rank of $u$ is out of the range of interest, the inclusions events for the neighbors of $u$ are independent if the graph has girth 7. Observe that this doesn't hold for black-box 2-rounds algorithms in general, not even in trees.}

~\\In this section we use the  higher inclusion probability of Algorithm~\ref{alg:weightedInclusion} to derive a small survival probability in graphs with girth at least $7$. %\seri{think about girth $7$ vs $8$ vs $6$.}\seri{rephrase here, make it clear, probably after the unifying of statements though}

\begin{lem}\label{lem:weighted-logstar}
Let $G$ be a graph with girth at least 7 and maximum degree $\Delta$, $\{p_v\}_{v\in V(G)}$ be nodes weights with $0\le p_v\le \tau = 1/\log\log\log\Delta$ for all $v\in V(G)$, and $d_v:=\sum_{u\sim v} p_u$. For a node $v\in V(G)$, let $L_v:=\{u\sim v\mid d_u\leq \ell\}$, where $\ell = (\log^*(1/\tau))^{1/100}$. We say that a node $v$ is good if the following two conditions hold:

\begin{enumerate}
    \item $d_v \ge 1/\ell$
    \item $\displaystyle \sum_{\scriptstyle u\in L_v} p_u \ge \frac{d_v}{100\ell^2}$

\end{enumerate}

It hols that algorithm~\ref{alg:weightedInclusion} finds an independent set $\mathcal{I}$ such that for each good node $v\in V(G)$,

$$\Pr[\exists u\in \mathcal{I}: u\sim v] \ge 1 - 2e^{-(\log^* \Delta)^{24/25}/10^{10}}$$
\end{lem}
\begin{proof}
We upper bound the complementary event. Notice that

\begin{align*}
\Pr[\forall u\sim v : u\notin \mathcal{I}] &= \Pr[\forall u\sim v : u\notin \mathcal{I} \text{ and } r_v \ge \tau^{1/10}] + \Pr[\forall u\sim v : u\notin \mathcal{I} \text{ and } r_v < \tau^{1/10}]\\
&\le \Pr[\forall u \in L_v : u\notin \mathcal{I} \mid r_v \ge \tau^{1/10}] + \Pr[r_v < \tau^{1/10}]\\
&= \prod_{u\in L_v} \Pr[u\notin \mathcal{I} \mid r_v \ge \tau^{1/10}] + \tau^{1/10}\\
&\le \prod_{u\in L_v} (1 - p_v\beta/10^7) + \tau^{1/10}\\
&\le e^{-d_v\beta/(\ell^2 10^9)} + \tau^{1/10}\\
&\le 2e^{-d_v\beta/(\ell^2 10^9)}\\
&\le 2e^{-(\log^*(1/\tau))/(\ell^4 10^9)}\\
&\le 2e^{-(\log^* \Delta)^{24/25}/10^{10}}
\end{align*}

where the second equality holds by conditional independence of the decisions of the neighbors of $v$ given the value of $v$ not being in any interval, which also uses the fact that the graph has girth at least 7. The second inequality follows from Lemma \ref{lem:better-inclusion-weighted-logstar}, and the third inequality follows from Proposition \ref{prop:exp-approx}.
\end{proof}

%\subsection{The Full Algorithm}

\subsection{The Full Algorithm for Proving Theorems~\ref{thm:main-girth} and~\ref{thm:trees}}\label{sec:full}
In this section, we present the full algorithm for proving Theorems~\ref{thm:main-girth} and~\ref{thm:trees}. The main idea is to combine our two-round algorithm from Lemma~\ref{lem:weighted-logstar} with Ghaffari's framework~\cite{Ghaffari16} (with some modifications). A brief overview of Ghaffari's algorithm is provided in Section~\ref{sec:TO}. Throughout this section, $\ell$ and $\tau$ are defined as in Lemma~\ref{lem:weighted-logstar}. 

~\\\noindent\textbf{Our Algorithm:} Our algorithm has a similar structure to Ghaffari's, with the following modifications: The weights $\{p^{(0)}_v\}_{v \in G(v)}$ are initialized to $\tau$ instead of $1/2$. Then, in each step $t$, we find an independent set $\mathcal{I}^{(t)}$ by using the two-round algorithm from Lemma~\ref{lem:weighted-logstar} (i.e., Algorithm~\ref{alg:weightedInclusion}). From round to round, the weights are updated as follows. Let $d_v^{(t)}:=\sum_{u\sim v} p_u^{(t)}$.

\[
p^{(t+1)}_v = \begin{cases}
p^{(t)}_v / \ell & \text{if } d^{(t)}_v > \ell,\\[8pt]
    \min\{\ell p_v^{(t)}, \tau\}, & \text{if } d^{(t)}_v \leq \ell
\end{cases}
\]

We show that after running this algorithm for $T=O(\log_{\ell} \Delta)$ rounds, the graph shatters into two parts, $H$ and $H'$. In $H$, the size of each connected component is small, and one can use deterministic algorithms to find an MIS in $H$ by using the same approach as in~\cite{Ghaffari16,BarenboimEPS16}. In $H'$, the maximum degree is $O(\ell / \tau)$. Therefore, after removing the nodes in the MIS of $H$ along with their neighbors, we can find an MIS in $H'$ efficiently by using a deterministic algorithm for low-degree graphs. This results in an MIS for the entire graph. 

~\\\noindent\textbf{Roadmap of the proofs:} The pre-shattering phase is given in Algorithm~\ref{alg:preshattering}, and the post-shattering phase is discussed in Lemmas~\ref{lem:golden-rounds},~\ref{lem:ruling-survival}, \ref{lem:shattering-small-components}, and~\ref{lem:deg-bound}. Before proving these lemmas, we introduce two key definitions. In Definition~\ref{def:golden}, we define the notion of \emph{golden iteration}, which is analogous to the \emph{golden rounds} terminology in~\cite{Ghaffari16}. By Lemma~\ref{lem:weighted-logstar}, a node survives with subconstant probability in a golden iteration. In Definition~\ref{def:deg}, we define the notion of \emph{degree reduction iteration}. Lemma~\ref{lem:golden-rounds} shows that each node either experiences $O(\log_{\ell}\Delta)$ golden iterations, or $O(\log_{\ell}\Delta)$ degree reduction iterations. 

Let $H$ be the graph induced by the set of nodes that experience $O(\log_{\ell}\Delta)$ golden iterations, and let $H'$ be the set of nodes induced by those who experience $O(\log_{\ell}\Delta)$ degree reduction iterations.  Lemma~\ref{lem:ruling-survival} implies that $H$ is shattered into small connected components. On the other hand, Lemma~\ref{lem:deg-bound} shows that $H'$ has maximum degree $O(\ell/\tau)$. We put everything together in the formal proof of Theorem~\ref{thm:main-girth}, which is provided right after the proof of Lemma~\ref{lem:deg-bound}. The proof of Theorem~\ref{thm:trees} is provided right after the proof of Theorem~\ref{thm:main-girth}.

%The full algorithm is given in Algorithm~\ref{}. Before proving its correctness\seri{...}

%\seri{understand better what happens with the small degree thing.}

%We run Mohsen's Algorithm, with two important modifications:

%\begin{enumerate}
 %   \item $p_v$s are multiplied and divided by $\ell$ rather than $2$, in order to speed up the algorithm.
  %  \item $p_v$s are capped at $\tau$ rather than $1/2$, in order to make it possible to apply Lemma \ref{lem:weighted-logstar}. To cope with this, we show that, after shattering, the remaining graph, rather than being empty, has maximum degree at most $O(1/\tau)$.
%\end{enumerate}

%We start by giving the algorithm:

\begin{algorithm}[H]
\caption{Pre-Shattering$(G)$}
\begin{algorithmic}[1]
\State Initialize $\mathcal{I} \gets \emptyset$, $G^{(0)}\gets G$, and $T \gets 10^6(\log \Delta)/(\log \ell)$.

\State Initialize $p^{(0)}_v\gets \tau$ for all $v\in V(G)$, and $d^{(0)}_v\gets \sum_{u\sim v} p^{(0)}_u$.

\For{$t \in [T]$}
    \State $\mathcal{I}^{(t)} \gets \text{Inclusion}(G^{(t-1)}, p^{(t-1)})$  \hspace{25pt} // \hspace{15pt} I.e., call Algorithm~\ref{alg:weightedInclusion}
    \State $\mathcal{I} \gets \mathcal{I} \cup \mathcal{I}^{(t)}$
    \State $G^{(t)} \gets G^{(t-1)}$ with all the nodes in $I^{(t)}$ and their neighbors deleted
    \For{all $v \in V(G^{(t)})$}
    \vspace{5pt}
      \State  \( 
p_v^{(t)} \gets \left\{ 
   \begin{array}{ll}
      p_v^{(t-1)}/\ell & \text{if } d_v^{(t-1)} > \ell, \\[5pt]
      \min(\ell p_v^{(t-1)}, \tau) & \text{if } d_v^{(t-1)} \leq \ell
   \end{array}
\right.
\)\vspace{5pt}
    \State $d^{(t)}_v\gets \sum_{u\sim v: u\in V(G^{(t)})} p^{(t)}_u$
    \EndFor
\EndFor

%\State $Y \gets \{ v \in V(G^{(T)}) : g_v^{(T)} \ge T / 50 \}$
%\State $H \gets G^{(T)}[Y]$ \Comment{Graph induced by nodes in $Y$}
%\State Find a $(10, 100\log\log n)$-ruling set $R$ of $H$
%\For{each $v \in R$}
 %   \State Let $C_v$ denote the set of vertices $u \in V(H)$ closest to $v$ in $H$, with ties broken arbitrarily
%\EndFor
%\State Let $H_R$ be the graph with vertex set $R$ and edges between $u, v \in R$ if $C_u$ and $C_v$ are adjacent in $H$
%\State Deterministically find a $100\log\log n$-radius, $100(\log\log n)^2$-color network decomposition of $H_R$, yielding a $10^4(\log\log n)^2$-radius, $100(\log\log n)^2$-color network decomposition of $H$
%\State Find an MIS $J$ of $H$ deterministically in $O((\log\log n)^4)$ time by processing color classes one at a time
%\State $I \gets I \cup J$
%\State Let $H' \gets G^{(T)}$ with all vertices in $J$ or adjacent to a vertex in $J$ deleted
%\State \Return $I \cup \text{Mohsen}(H')$
\State \Return $\mathcal{I}$
\end{algorithmic}\label{alg:preshattering}
\end{algorithm}

%In Mohsen's Algorithm, $h_v^{(T)}$ and $g_v^{(T)}$ count the number of type 1 and type 2 golden rounds respectively. In our analysis, we do not treat both types of golden rounds equally anymore, as type 1 golden rounds do not result in a low survival probability for $v$ anymore. Instead, we handle them implicitly at the end of the analysis by showing that the maximum degree of $H'$ is at most $O(1/\tau)$.

\begin{defn}\label{def:golden}\textbf{(Golden Iterations)}\\
    For a node $v\in V(G^{(t)})$, let $L^{(t)}_v$ be the set of neighbors $u\in V(G^{(t)})$ of $v$ for which $d^{(t)}_u\leq \ell$. We say that an iteration $t'$ of the for loop in Line 3 in Algorithm~\ref{alg:preshattering} is a golden iteration for a node $v$ if $d^{(t')}_v\geq 1/\ell$ and $\sum_{u\in L^{(t')}_v} p^{(t')}_u\geq d^{(t')}_v/(100\ell^2)$. We denote by $g^{(t)}_v$ be the number of golden iterations $t'\in [t]$ for $v$.
\end{defn}

%\If{$p_v^{(t)} = p_v^{(t-1)} = \tau$}
         %   \State $h_v^{(t)} \gets h_v^{(t-1)} + 1$
        %\Else
         %   \State $h_v^{(t)} \gets h_v^{(t-1)}$
        %\EndIf

\begin{defn}\label{def:deg}\textbf{(Degree Reduction Iterations)}\\
    We say that an iteration $t'$ is a degree reduction iteration for a node $v$ if $p^{(t')}_v = p^{(t'-1)}_v = \tau$. We denote by $h^{(t)}_v$ the number of degree reduction iterations $t'\in [t]$ for $v$.
\end{defn}

\begin{lem}\label{lem:golden-rounds}
Deterministically, for every vertex $v\in V(G^{(T)})$, $T \le h_v^{(T)} + 3g_v^{(T)} + 6\log_\ell \Delta$.
\end{lem}

\begin{proof}
%\aaron{Check $d_v^{(t)}$ superscripts carefully; make sure that they shouldn't be $t+1$ or $t-1$}. 
Let $N^{(t)}_v\subseteq V(G^{(t)})$ be the set of surviving neighbor of $v$ in the $t$'th iteration, and recall that $L^{(t)}_v$ is the set of neighbors $u\in N^{(t)}_v$ of $v$ for which $d^{(t)}_u\leq \ell$. First, we show that for every $t$ that is not a golden iteration  (i.e., $g_v^{(t)} = g_v^{(t-1)}$) and $d_v^{(t)} > 1/\ell$, $d_v^{(t+1)} \le 2d_v^{(t)}/\ell$. This is because, for such an iteration $t$, we have that:

\begin{align*}
d_v^{(t+1)} &= \sum_{u\in N^{(t+1)}_v} p_u^{(t+1)}\\
&\leq \sum_{u\in N^{(t)}_v} p_u^{(t+1)}\\
&= \sum_{u\in L_v^{(t)}} p_u^{(t+1)} + \sum_{u\in N^{(t)}_v\setminus L_v^{(t)}} p_u^{(t+1)}\\
&\le \sum_{u\in L^{(t)}_v} \ell p_u^{(t)} + \sum_{u\in N^{(t)}_v\setminus L_v^{(t)}} \frac{p_u^{(t)}}{\ell}\\
&< \ell\frac{d_v^{(t)}}{100\ell^2} + \frac{d_v^{(t)}}{\ell}\\
&\le \frac{2d_v^{(t)}}{\ell}
\end{align*}

For any $t$, the update rule for $p^{(t)}_v$ ensures that $d_v^{(t+1)} \le \ell d_v^{(t)}$. Hence, we can upper bound the number of $t\in [T]$ for which $d_v^{(t)} > \ell$, as follows. For any $t\in [T]$, let $e_v^{(t)}$ denote the number of $s\le t$ for which $d_v^{(s)} > \ell$. Consider two $a \le b\in \{2,3,\hdots,T\}$ with the property that $d_v^{(s)} > \ell$ for all $s$ with $a\le s\le b-1$, but with $d_v^{(a-1)} \le \ell$ as well. Observe that $d_v^{(a-1)} > 1/\ell$, so for all $s$ with $a-1 \le s \le b-1$ we have that:

$$d_v^{(s+1)} \le (2/\ell)^{1 - (g_v^{(s)} - g_v^{(s-1)})} \ell^{g_v^{(s)} - g_v^{(s-1)}} d_v^{(s)}$$ 

Thus,

$$\ell < d_v^{(b)} \le (2/\ell)^{b+1-a-(g_v^{(b-1)}-g_v^{(a-2)})}\ell^{g_v^{(b-1)}-g_v^{(a-2)}}d_v^{(a-1)} \le \ell^{(3/2)(g_v^{(b-1)}-g_v^{(a-2)}) - (b+1-a)/2}\ell$$

which implies that $b+1-a \le 3(g_v^{(b-1)}-g_v^{(a-2)})$. Furthermore, the number of iterations it takes to get out of the first interval $[0,\cdots, a]$ in which $d^{a'}_v> \ell$ for all $a'\in [a]$ is at most $\log_{\ell}(d^{0}_v)\leq \log_{\ell}(\Delta\cdot \tau)\leq 6\log_{\ell}\Delta$. Hence, summing over all possible intervals $[a,\cdots,b]$ for which $d^{a'}_v> \ell$ for every $a'\in [a]$ yields that $e_v^{(T)} \le 3g_v^{(T)} + 6\log_\ell \Delta$.

Next, we use the bound on $e_v^{(T)}$ to bound $T$. Each iteration counted by $e_v^{(T)}$ decreases $p_v^{(t)}$ by a factor of exactly $\ell$, while each iteration besides the ones counted by $h_v^{(T)}$ increases $p_v^{(t)}$ by a factor of exactly $\ell$. All other iterations do not change $p_v^{(t)}$. Thus,

$$\tau \ge p_v^{(T)} = \ell^{T - h_v^{(T)} - e_v^{(T)}}p_v^{(0)} = \ell^{T - h_v^{(T)} - e_v^{(T)}}\tau$$

Thus, $T - h_v^{(T)} - e_v^{(T)} \le 0$, which means that $T \le h_v^{(T)} + 3g_v^{(T)} + 6\log_\ell \Delta$ as desired.
\end{proof}

The following lemma is analogous to Lemma 4.1 in~\cite{Ghaffari16} and Lemma 3.2 in~\cite{BarenboimEPS16}. For completeness, we give a detailed proof in Appendix~\ref{app:diff}.

\begin{lem}\label{lem:ruling-survival}
Let $X$ be a set of vertices in $G$ for which the distance between any pair of vertices in $X$ within $G$ is at least 10. Then,

$$\Pr[\forall v\in X : v\in V(G^{(T)}) \text{ and } g_v^{(T)} \ge T/50] \le \gamma^{|X|T/50}$$

where $\gamma := 2e^{-(\log^*\Delta)^{24/25}/10^{10}}$.
\end{lem}

The following lemma is analogous to Lemma 4.2 in~\cite{Ghaffari16} and the discussion after it, and Lemma 3.3 and the discussion in Section 3.2 in~\cite{BarenboimEPS16}. For completeness, we give a full proof in Appendix~\ref{app:diff}.

\begin{lem}\label{lem:shattering-small-components}
    Let $Y: = \{ v \in V(G^{(T)}) : g_v^{(T)} \ge T / 50 \}$, and let $H$ be the graph induced by the nodes in $Y$. There is a $\poly(\log\log n)$-round algorithm that finds an MIS in $H$ with high probability.
\end{lem}

\begin{lem}\label{lem:deg-bound} Let $Y: = \{ v \in V(G^{(T)}) : g_v^{(T)} \ge T / 50 \}$. Deterministically the graph induced by 
$V(G^{(T)})\setminus Y$ has maximum degree $10\ell/\tau$.
\end{lem}

\begin{proof}
Suppose, for the sake of contradiction, that there exists $v\in V(G^{(T)})\setminus Y$ that has at least $z = 10^8\ell/\tau$ neighbors $u_1,u_2,\hdots,u_z$ for which $u_i \in V(G^{(T)})\setminus Y$ for all $i\in [z]$. By Lemma \ref{lem:golden-rounds}, $h_v^{(T)} \ge 9T/10$ and $h_{u_i}^{(T)} \ge 9T/10$ for all $i\in [z]$. We now will obtain contradictory bounds for the quantity

$$S = \sum_{t\in [T]:h_v^{(t)}\ne h_v^{(t-1)}} d_v^{(t-1)}$$

First, we obtain an upper bound. Recall that $h_v^{(t)} \ne h_v^{(t-1)}$ implies that $d_v^{(t-1)} \le \ell$ and $p_v^{(t)} = p_v^{(t-1)} = \tau$. By the first implication,

$$S \le \sum_{t\in [T]:h_v^{(t)}\ne h_v^{(t-1)}} \ell = h_v^{(T)}\ell \le T\ell$$

Next, we obtain a lower bound. For each $u_i$, there are at least $T - T/10 - T/10 = 4T/5$ values of $t$ for which both $h_v^{(t)}\ne h_v^{(t-1)}$ and $h_{u_i}^{(t)}\ne h_{u_i}^{(t-1)}$. Thus,

$$S \ge \sum_{i=1}^z \sum_{t\in [T]:h_v^{(t)}\ne h_v^{(t-1)} \text{ and } h_{u_i}^{(t)}\ne h_{u_i}^{(t-1)}} p_{u_i}^{(t-1)} \ge z (4T/5)\tau > T\ell$$

a contradiction, as desired.
\end{proof}

\begin{proof}[Proof of Theorem \ref{thm:main-girth}]
%The call to $\text{Mohsen}(H')$ at the end ensures that the returned independent set is a maximal independent set, so it suffices to bound the runtime and error probability. 

We run the pre-shattering algorithm (Algorithm~\ref{alg:preshattering}) for $T=O(\log_{\ell}\Delta)$ rounds. By Lemma~\ref{lem:golden-rounds}, each node either experiences $O(T)$ golden iterations or $O(T)$ degree reduction iterations. By Lemma~\ref{lem:shattering-small-components}, there is a $\poly(\log\log n)$-round algorithm that finds an MIS in the graph induced by the nodes who experience $O(T)$ golden iterations. By Lemma~\ref{lem:deg-bound} we can then use the deterministic $O(\Delta'+\log^*n)$ algorithm by~\cite{Kuhn09,BarenboimE09,BarenboimEK14}, where $\Delta'=O(\ell/\tau)\ll \poly(\log\log n)$, in the remaining graph induced by the nodes who experience $O(\log_{\ell}\Delta)$ degree reduction iterations, which takes $\ll \poly(\log\log n)$ rounds. 
\end{proof}

Before proving Theorem~\ref{thm:trees}, let us recall the following theorem from~\cite{BarenboimEPS16}.

\begin{thm}[Theorem 7.2 in~\cite{BarenboimEPS16}]\label{thm:deg-arb}
Let \( G \) be a graph of arboricity \( \lambda \) and maximum degree \( \Delta \), and let \( t \geq \max \{(5\lambda)^8, (4(c+1) \ln n)^7\} \) be a parameter. In \( O(\log_t \Delta) \) rounds, we can find an independent set \(\mathcal{I}\) such that, with probability at least \( 1 - n^{-c} \), the maximum degree in the graph induced by \( V(G) \setminus \mathcal{I} \) is at most \( t\lambda \).

\end{thm}

Now we are ready to prove Theorem~\ref{thm:trees}. Let us restate it for convenience.

\thmtrees*
\begin{proof}[Proof of Theorem~\ref{thm:trees}]
    We first use Theorem~\ref{thm:deg-arb} with $t=2^{\sqrt{\log n\cdot \log(\log^* n)}}$. Since the arboricity $\lambda=1$ in trees, this implies that the maximum degree of the tree becomes $\Delta'\leq 2^{\sqrt{\log n\cdot \log(\log^* n)}}$ in $O(\sqrt{\log n/\log(\log^* n)})$ rounds. Thus, by using Theorem~\ref{thm:main-girth}, we can find an MIS in the remaining tree in $O(\log(\Delta')/\log(\log^*\Delta') + \poly(\log\log n)) = O(\sqrt{\log n/\log(\log^* n)})$ rounds, as desired.
\end{proof}

\paragraph{Acknowledgment:} We would like to thank Amir Abboud, Sepehr Assadi, and Shafi Goldwasser for several helpful discussions and comments.
\bibliographystyle{plain}
\bibliography{refs.bib}

\appendix

\section{Standard Exponent Approximation}

\begin{prop}\label{prop:exp-approx}
For any $\delta > 0$ and $x\in (0,1-\delta)$,

$$e^{-x-x^2/\delta} \le 1 - x \le e^{-x}$$
\end{prop}

\begin{proof}
Apply $\ln$ to both sides. Recall that, for any $x\in (0,1)$,

$$\ln(1 - x) = -\sum_{i=1}^{\infty} \frac{x^i}{i}$$

Since $x > 0$, $\ln(1 - x) \le -x$, completing the upper bound. For the lower bound, use a geometric series upper bound. Specifically,

\begin{align*}
\ln(1 - x) &= -x - \frac{x^2}{2} - \sum_{i=3}^{\infty} \frac{x^i}{i}\\
&\ge -x - \frac{x^2}{2} - x^2 \sum_{i=3}^{\infty} (1 - \delta)^{i-2}\\
&= -x - \frac{x^2}{2} - x^2 (1 - \delta)/\delta\\
&\ge -x - \frac{x^2}{\delta}\\
\end{align*}

as desired.
\end{proof}

\section{Differed Proofs from Section~\ref{sec:full}}\label{app:diff}

We now give the proof of Lemma \ref{lem:ruling-survival}. In the proof, we use the following notation:  For a tuple of sets $(S^{(a)})_{a\in I}$ for some integer $a$ and set of integers $I$, define $|S| := \sum_{a\in I} |S^{(a)}|$.

\begin{proof}[Proof of Lemma~\ref{lem:ruling-survival}]
For each $t\in [T]$, define a set-valued random variable $W^{(t)}\subseteq X$ to be the set of all $v\in X$ for which both $t$ is a golden iteration for $v$ (i.e., $g_v^{(t)} = g_v^{(t-1)} + 1$) and $g_v^{(t)} \le T/50$. Let $\mathcal{Z}$ denote the family of all $T$-tuples of sets $(Z^{(t)})_{t=1}^T$, where $Z^{(t)} \subseteq X$ for all $t\in [T]$ and each member of $X$ appears in exactly $T/50$ different sets $Z^{(t)}$. Since choices of $(Z^{(t)})_{t=1}^T$ for $(W^{(t)})_{t=1}^T$ are mutually exclusive,

\begin{align*}
    &\Pr[\forall v\in X : v\in V(G^{(T)}) \text{ and } g_v^{(T)} \ge T/50]\\
    &= \sum_{(Z^{(t)})_{t=1}^T \in \mathcal{Z}} \Pr[(\forall v\in X : v\in V(G^{(T)})) \text{ and } (\forall t\in [T] : W^{(t)} = Z^{(t)})]
\end{align*}

For each $t\in \{0,1,2,\hdots,T-1,T\}$, let $\mathcal{P}^{(t)}$ denote the family of all $t$-tuples of sets $P^{(s)} \subseteq X$ for which each member of $X$ appears in at most $T/50$ different sets $P^{(s)}$. $\mathcal{P}^{(0)}$ is defined to be the singleton family consisting just of the empty tuple. We now show for all $t\in \{0,1,\hdots,T\}$ that

\begin{align*}
    &\Pr[\forall v\in X : v\in V(G^{(T)}) \text{ and } g_v^{(T)} \ge T/50]\\
    &\le \sum_{(P^{(s)})_{s=1}^t \in \mathcal{P}^{(t)}} \gamma^{|X|T/50 - |P|}\Pr[(\forall v\in X : v\in V(G^{(t)})) \text{ and } (\forall s\in [t] : W^{(s)} = P^{(s)})]
\end{align*}

We prove this by induction in descending order of $t$. For $t = T$, the result holds trivially because $\mathcal{Z} \subseteq \mathcal{P}^{(T)}$. For the inductive step, consider some $t \in \{0,1,\hdots,T-1\}$ and assume that the result holds for $t+1$. In particular, we are given that

\begin{align*}
&\Pr[\forall v\in X : v\in V(G^{(T)}) \text{ and } g_v^{(T)} \ge T/50]\\
&\le \sum_{(P^{(s)})_{s=1}^{t+1} \in \mathcal{P}^{(t+1)}} \gamma^{|X|T/50 - |P|}\Pr[(\forall v\in X : v\in V(G^{(t+1)})) \text{ and } (\forall s\in [t+1] : W^{(s)} = P^{(s)})]\\
&= \sum_{(P^{(s)})_{s=1}^{t+1} \in \mathcal{P}^{(t+1)}} \gamma^{|X|T/50 - |P|}\Pr[(\forall v\in X : v\in V(G^{(t+1)})) \text{ and } W^{(t+1)} = P^{(t+1)} \text{ and } (\forall s\in [t] : W^{(s)} = P^{(s)})]\\
\end{align*}

$W^{(s)}$ for all $s\in [t+1]$ is a function of the randomness only from the first $t$ iterations of the for loop. $G^{(t)}$ is also only a function of the randomness from the first $t$ iterations. Given $G^{(t)}$, the event $\{\forall v\in X : v\in V(G^{(t+1)})\}$, i.e. survival of all $X$ in iteration $t+1$, is only a function of the randomness from the $t+1$th iteration of the for loop. Recall that $\text{Inclusion}$ is the only source of randomness in the algorithm and that survival within round $t+1$ only depends on randomness in the 3-neighborhood of a given vertex. Thus, being distance $\ge 10$ apart ensures that neighborhoods are disjoint and that the events $\{v\in V(G^{(t+1)})\}$ are conditionally independent given $G^{(t)}$. Therefore,

\begin{align*}
&\Pr[(\forall v\in X : v\in V(G^{(t+1)})) \mid W^{(t+1)} = P^{(t+1)} \text{ and } (\forall s\in [t] : W^{(s)} = P^{(s)})]\\
&= \E_{G^{(t)}}\left[\Pr\left[(\forall v\in X : v\in V(G^{(t+1)})) \mid G^{(t)}\right] \mid (\forall s\in [t+1] : W^{(s)} = P^{(s)})\right]\\
&= \E_{G^{(t)}}\left[\prod_{v\in X}\Pr\left[v\in V(G^{(t+1)}) \mid G^{(t)}\right] \mid (\forall s\in [t+1] : W^{(s)} = P^{(s)})\right]\\
&\le \E_{G^{(t)}}\left[\prod_{v\in P^{(t+1)}}\Pr\left[v\in V(G^{(t+1)}) \mid G^{(t)}\right] \mid (\forall s\in [t+1] : W^{(s)} = P^{(s)})\right]\\
&\le \E_{G^{(t)}}\left[\gamma^{|P^{(t+1)}|} \mid (\forall s\in [t+1] : W^{(s)} = P^{(s)})\right]\\
&= \gamma^{|P^{(t+1)}|}
\end{align*}

Since $W^{(t+1)}$ is defined if and only if $X\subseteq V(G^{(t)})$,

\begin{align*}
&\sum_{(P^{(s)})_{s=1}^{t+1} \in \mathcal{P}^{(t+1)}} \gamma^{|X|T/50 - |P|}\Pr[(\forall v\in X : v\in V(G^{(t+1)})) \text{ and } W^{(t+1)} = P^{(t+1)} \text{ and } (\forall s\in [t] : W^{(s)} = P^{(s)})]\\
&\le \sum_{(P^{(s)})_{s=1}^{t+1} \in \mathcal{P}^{(t+1)}} \gamma^{|X|T/50 - |P| + |P^{(t+1)}|}\Pr[W^{(t+1)} = P^{(t+1)} \text{ and } (\forall s\in [t] : W^{(s)} = P^{(s)})]\\
&= \sum_{(P^{(s)})_{s=1}^t \in \mathcal{P}^{(t)}} \gamma^{|X|T/50 - |P|}\Pr[(\forall v\in X : v\in V(G^{(t)})) \text{ and } (\forall s\in [t] : W^{(s)} = P^{(s)})]\\
\end{align*}

completing the induction step, where the last equality uses the exclusivity of choices for $P^{(t+1)}$. Thus, plugging in $t = 0$ shows that

$$\Pr[\forall v\in X : v\in V(G^{(T)}) \text{ and } g_v^{(T)} \ge T/50] \le \gamma^{|X|T/50}$$

as desired.
\end{proof}

Next, we prove Lemma \ref{lem:shattering-small-components}. %This is the same proof as the one given for Lemma \aaron{Cite Ghaffari}.

\begin{proof}[Proof of Lemma~\ref{lem:shattering-small-components}]
Given a decomposition of a graph into clusters $\{C_i\}_{i}$, the clustering graph is defined to be the graph where each cluster $C_i$ is a node, and each two nodes $C_i$ and $C_j$ are connected if there is an edge from some node $u\in C_i$ to some node $v\in C_j$ in the original graph. An $(r,c)$ network decomposition is a decomposition of the graph into clusters where the radius of each cluster is at most $r$ and clustering graph has chromatic number at most $c$. Given an $(r,c)$ network decomposition for $H$, one can find an MIS in $H$ in $O(r\cdot c)$ rounds by going over the color classes one at a time, and each time find an MIS for the clusters of the current color class (after deleting all the nodes found in previous independent sets together with their neighbors).

Hence, to prove the lemma, we show that $H$ admits an $\left(O(\log\log^2 n), O(\log\log n)\right)$ network decomposition that can be found efficiently. For this, we find a $(10, 100\log\log n)$-ruling set $R$ of $H$ by using the algorithm of~\cite{GfellerV07,SchneiderEW13}, which takes $O(\log\log n)$ rounds.\footnote{An $(\alpha,\beta)$-ruling set is a set of nodes where any two nodes in the set are at distance at least $\alpha$ from each other, and any other node is at distance at most $\beta$ from some node in the ruling set.} Then, we make a cluster $C_v$ around each node $v\in R$, where each node $u \in V(H)$ joins the cluster of the closest $v\in R$, with ties broken arbitrarily. Let $H_R$ be the clustering graph corresponding to the clusters $\{C_v\}_{v\in R}$. To show that $H$ admits an $\left(O(\log\log^2 n), O(\log\log n)\right)$ network decomposition, it suffices to show that $H_R$ admits an $(O(\log\log n),O(\log\log n))$-network decomposition (and that it can be found efficiently together with the coloring of the clusters). Indeed, given such an $(O(\log\log n),O(\log\log n))$-network decomposition $\mathcal{C}_R$ for $H_R$ we can obtain a clustering $\mathcal{C}$ for $H$ as follows. Recall that each cluster in $\mathcal{C}_R$ is a union of riling-set clusters. For each $C'\in \mathcal{C}_R$, we define a cluster $C$ in $H$ that contains all the nodes in the union of all the ruling-set clusters in $C$. This implies an $O(\log\log n)^2$-radius, $O(\log\log n)$-color network decomposition $\mathcal{C}$ of $H$.

Thus, it remains to show that $H_R$ admits the desired $(O(\log\log n),O(\log\log n))$-network decomposition, and that it can be found efficiently. Recent algorithms show that in a graph with $n'$ nodes, an $(O(\log n'),O(\log n'))$-network decomposition can be found in $\tilde{O}(\log^2(n'))$ rounds deterministically.\footnote{Prior algorithms already imply $(O(\polylog n'),O(\polylog n'))$-network decomposition, which would suffice for our purposes. The very recent algorithm by~\cite{ecomposition} improves the polynomial in our $\poly(\log\log n)$ runtime for this lemma.} Hence, it suffices to show that each connected component of $H_R$ has at most $100\log_{\Delta} n$ vertices with high probability. Suppose, for the sake of contradiction, that there is a connected set $X\subseteq R$ with $|X| \ge 100\log_{\Delta} n$ in $H_R$. By construction of $H_R$, these vertices are also connected in $G^{(T)}$. We construct a set of nodes $Y\subseteq V(G^{(T)})$ as follows. First, we add all the nodes in $X$ to $Y$. Then, while there exists a vertex in $V(G^{(T)})$ with $G$-distance at least $10$ from some node in $Y$, we add this vertex to $Y$. The connectedness of $X$ in $G^{(T)}$ implies that any vertex in $Y$ is within $G$-distance 21 of some other vertex in $Y$. However, no such sets exist with high probability, because the number of such sets in $G$ is at most $n\Delta^{21|X|}$ and a union bound with Lemma \ref{lem:ruling-survival} implies that such a set exists with probability at most

$$\gamma^{|Y|T/50}n\Delta^{21|Y|} \le n\Delta^{-|Y|} \le n\Delta^{-|X|} \le n^{-50}$$

Thus, each connected component of $H_R$ has at most $O(\log_{\Delta} n)$ with high probability, as desired. 

\end{proof}

\section{Simple One-Round Lower Bound}\label{sec:lb}

In this section, we demonstrate the fundamental nature of the constant survival probability barrier when using previous approaches. As discussed in Section~\ref{sec:TO}, all randomized MIS algorithms in the literature, as well as ours, do not rely on vertex identifiers.

\begin{thm}\label{thm:lb}
    Let $G$ be a $\Delta$-regular graph without any constant-length cycles. Any randomized distributed one-round algorithm with no vertex IDs that always returns an independent set $\mathcal{I}$ has the property that each node is in $\mathcal{I}$ or is adjacent to a node $\mathcal{I}$ with probability at most $p=1 - 1/(2e) + 1/\Delta$.
\end{thm}

This hardness result shows that all 1-round algorithms for independent set have constant vertex survival probability. This is in contrast with Lemma \ref{lem:weighted-logstar}, which shows that there is a 2-round algorithm for independent set with subconstant vertex survival probability.

To prove Theorem \ref{thm:lb}, we start viewing 1-round no-vertex-ID LOCAL algorithms for independent set in a simpler way. For any integer $k$, let $\mathcal{S}_k$ be the family of sets of numbers with size exactly $k$. A family of sets $\mathcal{V}\subseteq [0,1]\times \mathcal{S}_{\Delta}$ is called an \emph{independence family} if an only if for any set $S\in \mathcal{S}_{\Delta+1}$ and any distinct $u,v\in S$, either $(v,S\setminus \{v\})\notin \mathcal{V}$ or $(u,S\setminus \{u\})\notin \mathcal{V}$. For an independence family $\mathcal{V}$, let

$$P_{\mathcal{V}} := \Pr_{x\sim [0,1], \forall i\in [\Delta]: y_i\sim [0,1], S_i\sim \mathcal{S}_{\Delta-1}}[\forall i\in [\Delta] : (y_i,S_i\cup \{x\})\notin \mathcal{V} \text{ and } (x,\{y_1,y_2,\hdots,y_{\Delta}\})\notin \mathcal{V}]$$

In this probability, $x$ and all $y_i$s are sampled uniformly from $[0,1]$, and $S_i$ is sampled uniformly from $\mathcal{S}_{\Delta-1}$, all independently. This corresponds with sampling all randomness in a 2-neighborhood and assessing the probability that the vertex with rank $x$ survives. We now show the following:

\begin{prop}\label{prop:local-independence-equiv}
    All 1-round no-vertex-ID randomized LOCAL algorithms for independent set have vertex survival probability at least $P_{\mathcal{V}^*}$, where $\mathcal{V}^*$ is an independence family that minimizes $P_{\mathcal{V}^*}$.
\end{prop}

\begin{proof}
    For every 1-round no-vertex-ID LOCAL algorithm $A$, there is an associated independence family $\mathcal{V}_A$ of 1-neighborhood views that lead to inclusion in the independent set. This family is an independence family because of the fact that $A$ always outputs an independent set, and $(u,S\setminus \{u\})$ and $(v,S\setminus \{v\})$ can be adjacent to each other in a $\Delta$-regular graph that has no constant-length cycles. $P_{\mathcal{V}_A}$ is the survival probability of a vertex (with rank $x$) when the algorithm $A$ is applied, as randomness tapes for different vertices are all distinct with probability 1. Minimizing over all valid algorithms $A$ gives the desired result.
\end{proof}

Next, we upper bound the inclusion probability of a 1-round algorithm, where we fix one rank as done in Lemma \ref{lem:weighted-logstar}:

\begin{prop}\label{prop:hardness-inclusion}
    For any $x\in [0,1]$ and any independence family $\mathcal{V}$,

    $$\Pr_{y\sim [0,1], S\sim \mathcal{S}_{\Delta-1}}[(y,S\cup \{x\}) \in \mathcal{V}] \le \frac{1}{\Delta}$$
\end{prop}

\begin{proof}
    We can alternatively sample the pair $(y,S)$ by (a) sampling $T$ uniformly from $\mathcal{S}_{\Delta}$, (b) sampling $y$ uniformly from $T$, and (c) defining $S := T\setminus \{y\}$. Since $\mathcal{V}$ is an independence family, there exists at most one $z\in T$ for which $(z, T\cup \{x\}\setminus \{z\})\in \mathcal{V}$. Since $y$ is sampled uniformly,

    $$\Pr_{y\sim T_0}[(y,S\cup \{x\}) \in \mathcal{V} \mid T = T_0] \le \frac{1}{\Delta}$$

    for any $T_0\in \mathcal{S}_{\Delta}$. Taking the expectation over all possible $T_0$ shows that

    $$\Pr_{T,y}[(y,S\cup \{x\}) \in \mathcal{V}] \le \frac{1}{\Delta}$$

    as desired.
\end{proof}

We can now use independence to get the desired lower bound on survival probability:

\begin{proof}[Proof of Theorem \ref{thm:lb}]
Define $\mathcal{V}^*$ as in Proposition \ref{prop:local-independence-equiv}. For any $x\in [0,1]$,

$$\Pr_{\forall i\in [\Delta] : y_i\sim [0,1], S_i\sim \mathcal{S}_{\Delta-1}}[\forall i\in [\Delta] : (y_i, S_i)\notin \mathcal{V}^*] \ge (1 - 1/\Delta)^{\Delta} \ge 1/(2e)$$

by independence and Proposition \ref{prop:hardness-inclusion}. Taking the expectation over $x$ shows that

$$\Pr_{x\sim [0,1], \forall i\in [\Delta] : y_i\sim [0,1], S_i\sim \mathcal{S}_{\Delta-1}}[\forall i\in [\Delta] : (y_i, S_i)\notin \mathcal{V}^*] \ge 1/(2e)$$

Taking the expectation over $x$ in Proposition \ref{prop:hardness-inclusion} and replacing $y$ from the proposition with $x$ from this proof shows that

$$\Pr_{x\sim [0,1], \forall i\in [\Delta] : y_i\sim [0,1]}[(x, \{y_1,y_2,\hdots,y_{\Delta}\})\in \mathcal{V}^*] \le \frac{1}{\Delta}$$

Therefore, by a union bound, $P_{\mathcal{V^*}}\ge 1/(2e) - 1/\Delta$. Proposition \ref{prop:local-independence-equiv} yields the desired result.
\end{proof}

\end{document}